\documentclass[aip,jmp,reprint,onecolumn,notitlepage,nofootinbib]{revtex4-1}
\usepackage{hyperref}
\usepackage{amsmath}
\usepackage{amssymb}
\usepackage{fullpage}
\usepackage{bbold}
\usepackage{amsthm}
\usepackage{graphicx}
\usepackage{cleveref}
\usepackage{dsfont}
\usepackage{xcolor}
\usepackage{tikz}

\DeclareMathOperator{\tr}{Tr}

\newcommand{\id}{\mathbb{1}}
\newcommand{\inp}[2]{\langle#1,#2\rangle}

\newcommand{\mc}[1]{\mathcal{#1}}
\newcommand{\bs}[1]{\boldsymbol{#1}}
\newcommand{\bc}[1]{\boldsymbol{\mathcal{#1}}}
\newcommand{\md}[1]{\mathds{#1}}

\newcommand{\ct}{^\dagger}
\newcommand{\tn}[1]{^{\otimes #1}}

\newcommand{\C}{\mc{C}}

\newtheorem{theorem}{Theorem}
\newtheorem{definition}{Definition}
\newtheorem{lemma}{Lemma}

\newcounter{notecounter}

\DeclareMathOperator{\vsp}{span}
\usepackage{appendix}

\newcommand{\Vasc}{V_{[\mathrm{A}]}}
\newcommand{\Vasa}{V_{\{\mathrm{A}\}}}
\newcommand{\Vsc}{V_{[\mathrm{S}]}}
\newcommand{\Vsa}{V_{\{\mathrm{S}\}}}
\newcommand{\Vadjc}{V_{[\mathrm{adj}]}}
\newcommand{\Vadja}{V_{\{\mathrm{adj}\}}}
\newcommand{\Vadjap}{V^\perp_{\{\mathrm{adj}\}}}

\newcommand{\bsq}{\bs{\sigma}_q}
\newcommand{\bsqh}{\bs{\hat{\sigma}}_q}
\usepackage{natbib}
\usepackage{epstopdf}
\usepackage{tikz-cd}
\usepackage{array}
\usepackage{moresize}
\newcolumntype{C}{>{$\displaystyle} c <{$}}

\begin{document}
\title{Representations of the multi-qubit Clifford group}
\author{Jonas Helsen}
\affiliation{QuTech, Delft University of Technology, Lorentzweg 1, 2628 CJ Delft, The Netherlands}
\author{Joel J. Wallman}
\affiliation{Institute for Quantum Computing, University of Waterloo, Waterloo, Ontario N2L 3G1, Canada}
\affiliation{Department of Applied Mathematics, University of Waterloo, Waterloo, Ontario N2L 3G1, Canada}
\author{Stephanie Wehner}
\affiliation{QuTech, Delft University of Technology, Lorentzweg 1, 2628 CJ Delft, The Netherlands}
\date{\today}
\begin{abstract}
The $q$-qubit Clifford group, that is, the normalizer of the $q$-qubit Pauli
group in $U(2^q)$, is a fundamental structure in quantum information with a wide
variety of applications. We characterize all irreducible subrepresentations of
the two-copy representation $\varphi\tn{2}$ of the Clifford group on the
two-fold tensor product of the space of linear operators
$\mathcal{M}_{2^q}^{\otimes 2}$. In a companion paper [Helsen et al. arXiv:1701.04299 (2017)] we apply
this result to improve the statistics of randomized benchmarking, a method for
characterizing quantum systems.
\end{abstract}

\maketitle
\section{Introduction}
Symmetric structures, encoded as groups, play a fundamental role in the study 
of quantum information theory and quantum mechanics in general. The Pauli group
and its normalizer, the Clifford group, are particularly important in quantum
information, with applications such as quantum error-correcting
codes~\cite{Gottesman2010}, quantum tomographic methods~\cite{Gross2010a}, and
quantum data hiding~\cite{DiVincenzo2002}. Furthermore, operations within the
Clifford group can be efficiently simulated~\cite{Aaronson2004} and the Clifford
group is a unitary $2$-design \cite{DiVincenzo2002}, that is, averages over the
standard representation of the Clifford group reproduce the first
two moments of the Haar average over the full unitary group~\cite{Dankert2009}.
These properties make the Clifford group useful for characterization protocols 
for quantum systems such as randomized benchmarking~\cite{Magesan2011}.\\

A subgroup of the unitary group (in our case the Clifford group) is a unitary $t$-design if the
irreducible subrepresentations of $t$ tensor copies of its standard representation are in one-to-one correspondence
with the irreducible subrepresentations of the same construction involving the 
the full unitary group~\cite{Gross2007a}. This equivalent definition is useful
because the tensor representations of the unitary group are well understood via
Schur-Weyl duality~\cite{Goodman2009}.\\

Recently it has been shown that the $q$-qubit Clifford group is also a unitary
$3$-design~\cite{Zhu2015,Webb2015}. However, simultaneously it was shown that
the multi-qubit Clifford group is not a unitary $4$-design. Consequently, the
representation of $4$ tensor copies of the standard representation of the Clifford group differs from the same construction using
the unitary group. In this paper we will analyze a closely related representation of the Clifford group which we call the two-copy representation. This representation is the tensor product of two tensor copies of the standard representation and two tensor copies of the dual of the standard representation. The structure of the two-copy representation of the single-qubit
Clifford group was analyzed in~\cite{Wallman2014} and used to analyze the
statistical performance of randomized benchmarking.\\

In this paper we provide a complete analysis of the two-copy representation
of the multi-qubit Clifford group for any number of qubits. In a companion 
paper~\cite{Helsen2016}, we use these results to analyze multi-qubit 
randomized benchmarking, leading to a substantial reduction in the amount of 
data required to obtain rigorous and precise estimates using the randomized benchmarking procedure. \\

\section{Preliminaries}
We begin by setting some relevant notation. We denote by $\mc{M}_d$ the vector
space of linear operators from $\md{C}^d$ to $\md{C}^d$. We will only be 
interested in the case $d= 2^q$ where $q\in \md{N}$ is the number of qubits in 
the system. Other finite dimensional vector spaces will often be denoted by $V$, with $|V|$ the dimension of that vector space. We also recall that the vector space $\mc{M}_d$ can be equipped with
the Hilbert-Schmidt inner product which takes the form
\begin{equation}\label{eq:Hilbert-Schmidt}
\inp{A}{B} = \tr(A\ct B) \hspace{5mm}\forall A,B \in \mc{M}_d.
\end{equation}
We also denote the (anti-) commutator of two elements of $\mc{M}_d$ as
\begin{align*}
[A,B] = AB-BA \hspace{5mm} \forall A,B \in \mc{M}_d \tag{commutator}\\
\{A,B\} = AB+BA \hspace{5mm} \forall A,B \in \mc{M}_d \tag{anti-commutator}
\end{align*}
To facilitate later analysis we also recall some standard facts about 
representation and character theory. We will mostly follow~\cite{Fulton2004}. Readers familiar with representation theory may skip
this part.
\subsection{Representation theory}
Let $GL(V)$ be the general linear group over a finite dimensional
(real or complex) vector space $V$, i.e. the group of invertible linear
transformations of $V$. Let $G$ be a finite group. A representation $\varphi_V$ of $G$
on a finite dimensional (real or complex) vector space $V$ is a map
\begin{equation}
\varphi_V:G\to GL(V): g \mapsto \varphi(g)
\end{equation}
with the property
\begin{equation}
\varphi_V(g)\varphi_V(h) = \varphi_V(gh) \hspace{5mm}\forall g,h \in G.
\end{equation}
We call $V$ the space carrying the representation $\varphi_V$. A subspace $W\subset V$ carries a subrepresentation of $\varphi_V$ (denoted $\varphi_W$) if
\begin{equation}
\varphi_V(g)W \subset W
\end{equation}
for all $g\in G$. A representation $\varphi_V$ is called \emph{irreducible} if
there is no non-trivial $(W\neq0)$ proper subspace $W$ of $V$ such that $\varphi_V(g) W\subset W$
for all $g\in G$. Two representations $\varphi_V$ and $\varphi_{V'}$ are called
equivalent, denoted $\varphi\cong \varphi'$, if and only if there exists an invertible linear map  $T:V'\to V$ such that
\begin{equation}
\varphi_V(g) = T \varphi_{V'}(g) T^{-1} \hspace{5mm}\forall g\in G.
\end{equation}
A central result for irreducible representations is Schur's
lemma~\cite{Fulton2004}. Let $\varphi_V,\varphi_{V'}$ be irreducible representations
of a finite group $G$ on spaces $V, V'$. Then Schur's lemma states that a linear map
$A:V\to V'$ satisfies
\begin{equation}\label{eq:commutation}
\varphi_{V'}(g)A = A \varphi_{V}(g) \hspace{5mm}\forall g\in G,
\end{equation}
if and only if $A$ is of the form
\begin{equation}
A = \begin{cases}{} 0 \;\;\:\:\hspace{0.9mm} \text{if}\;\; \varphi_V\not \cong\varphi_{V'},\\
\lambda \id \;\;\text{if}\;\;  \varphi_V \cong\varphi_{V'}.\end{cases}
 \end{equation}
for some $\lambda\in \md{C}$. We note an important corollary. Let $\varphi_V = \oplus_i \varphi_i$ be a
representation of a group $G$ composed of irreducible inequivalent 
representations $\varphi_i$ and let $A:V\to V$ satisfy $\varphi(g) A = A \varphi(g)$ for all $g \in G$. Then we must have
\begin{equation}\label{eq:Schur corrolary}
A = \sum_i \lambda_i P_i
\end{equation}
for some $\lambda_i\in \md{C}$, where the $P_i$ are projectors onto the subspaces of $V$
carrying the irreducible subrepresentations $\varphi_i$.

We next recall the character of a representation. Let $\varphi: G\to V$
be a representation of a finite group $G$ on a finite dimensional (real or
complex) vector space $V$. The character $\chi_\varphi: G\to \md{C} $ of
the representation $\varphi$ is defined as
\begin{align}
\chi_{\varphi}: &G\to\md{C}:g\mapsto  \chi_{\varphi}(g) = \tr_V(\varphi(g)),
\end{align}
where $\tr_V(\;)$ denotes the trace over the vector space $V$.
Characters have a number of useful properties~\cite{Fulton2004} which we recall here.
For representations $\varphi,\varphi'$ we have the relations
\begin{align}
\chi_{\varphi\otimes \varphi'} &= \chi_\varphi\; \chi_{\varphi'},\\
\chi_{\varphi\oplus \varphi'} &= \chi_\varphi +\chi_{\varphi'},
\end{align}
with suitable generalizations to multiple direct sums and tensor products.
The inner product between the characters of two representations $\varphi$ and $\varphi'$ of 
a finite group $G$ is
\begin{equation}
\inp{\chi_\varphi}{\chi_{\varphi'}}
:= \frac{1}{|G|}\sum_{g\in G} \chi_{\varphi}(g)^*\chi_{\varphi'}(g),
\end{equation}
where $*$ denotes complex conjugation. Schur's orthogonality relations state
that for \emph{irreducible} representations $\varphi, \varphi'$ of $G$,
\begin{equation}\label{eq:SchurOrthogonality}
\inp{\chi_\varphi}{\chi_{\varphi'}} = \begin{cases} 0 \hspace{2mm}\text{if} \;\;\varphi\not\cong \varphi'\\
1 \hspace{2mm}\text{if} \;\;\varphi\cong \varphi'.\end{cases}
\end{equation}
We note the following useful corollary. 
Let $\varphi = \oplus_{i}\varphi_i^{\oplus n_i}$ where the $\varphi_i$ are all 
irreducible, inequivalent representations and $n_i$ denotes the multiplicity of
$\varphi_i$ in $\varphi$. Then \cref{eq:SchurOrthogonality} implies
\begin{equation}\label{Schur}
\inp{\chi_\varphi}{\chi_\varphi} = \sum_{i}\;\; n_i^2.
\end{equation}
In particular we also have,
\begin{equation}\label{eq:Schur_ineq}
\inp{\chi_\varphi}{\chi_\varphi} \geq 1,
\end{equation}
with equality if and only if $\varphi$ is irreducible.

\subsection{The Pauli and Clifford groups}

Finally we recall definitions for the Pauli and Clifford groups, note some useful facts about the Pauli group and define what we mean by the ``two-copy representation'' of the Clifford group.

\begin{definition}[Multi-qubit Pauli and Clifford groups]
Take $U(d)$ to be the group of $d\times d$ unitary matrices (where $d=2^q)$, which has a standard representation~\cite{Fulton2004} on the complex vector space $\md{C}^d $. For $q=1$, let $\{v_0,v_1\}$ be an orthonormal basis of $\md{C}^2$ and in this basis define the following linear operators by their action on the basis
\begin{align*}
X v_l = v_{l+1},\ 
Z v_l =(-1)^lv_l,\
Yv_l = iZXv_l = i(-1)^{l+1}v_{l+1},
\end{align*}
for $l\in \{0,1\}$ and addition over indices is taken modulo $2$. Note that $X,Y,Z \in U(2)$. The $q$-qubit Pauli group $\mc{P}_q$ is now defined as the subgroup of the unitary group $\mathrm{U}(2^q)$ 
consisting of all $q$-fold tensor products of $q$ elements of $\mc{P}_1:=\langle X,Z,i\id_2 \rangle$. The $q$-qubit Clifford group $\mc{C}_q$ is the normalizer (up to complex phases) of 
$\mc{P}_q$ in $\mathrm{U}(2^q)$, that is,
\begin{align*}
\mc{C}_q = \{U\in\mathrm{U}(2^q)\;\;:\;\; U\mathcal{P}_q U\ct \subseteq \mathcal{P}_q \}/U(1).
\end{align*}
Both the Clifford and Pauli groups have a standard faithful representation on the vector space $\md{C}^d$ with $d=2^q$ as they are subgroups of the unitary group $U(d)$. We also define $\hat{\mc{P}}_q$ as the subset of $\mc{P}_q$ consisting of all $q$-fold tensor products of the operators $\{\id,X, Y, Z\}$. Note that all the elements of $\hat{\mc{P}}_q$ are Hermitian.
\end{definition}

For a more expansive introduction to the Pauli and Clifford groups see e.g.~\cite{Farinholt2014} and references therein.\\

Next we recall some useful facts about the Pauli group and its standard representation. We begin by noting that the Hermitian subset $\hat{\mc{P}}_q$ of the
Pauli group forms an orthogonal basis for the Hilbert space
$\mc{M}_d$. We can turn this into an orthonormal basis (under
the Hilbert-Schmidt inner product in \cref{eq:Hilbert-Schmidt}) by introducing
normalized Pauli matrices :
\begin{equation}
\sigma_0 = \frac{\id}{\sqrt{d}},\hspace{3mm} \bsq = \left\{\frac{P}{\sqrt{d}}\;\;\|\;\; P\in \hat{\mc{P}}_q \backslash \{\id\}\right\},
\end{equation}
where we have given the normalized identity its own symbol for later convenience. We also define $\bsqh := \bsq\cup\{\sigma_0\}$. We will denote the elements of the set $\bsq$ by Greek letters ($\sigma, \tau, \nu,...$). For the case of a single qubit we denote the normalized $X,Y,Z$ Pauli matrices by $\sigma_X,\sigma_Y,\sigma_Z$. We also, for later convenience, introduce the \emph{normalized} matrix product of two normalized Pauli matrices as
\begin{equation}
\sigma\cdot \tau  := \sqrt{d} \sigma\tau \hspace{5mm}\sigma,\tau\in \bsqh.
\end{equation}
Note that $\sigma\cdot \tau \in \pm\bsqh$ if $[\sigma, \tau]=0$ and $i\sigma\cdot \tau \in \pm\bsq$ if $\{\sigma, \tau\}=0$.
Lastly we define the following parametrized subsets of $\bsq$ and $\bsqh$. For all $\tau
\in\bsq$ we define
\begin{align}
\bs{N}_\tau:= \{\sigma \in \bsq \;\;\|\;\; \{\sigma,\tau\}=0\},\\
\bs{C}_\tau:= \{\sigma \in \bsq\backslash\{\tau\} \;\;\|\;\; [\sigma,\tau]=0\},\\
\bs{\hat{C}}_\tau := \{\sigma \in \bsqh \;\;\|\;\; [\sigma,\tau]=0\}.
\end{align}
Note that we have $|\bs{\hat{C}}_\tau| = |\bs{N}_\tau| = \frac{d^2}{2}$ and 
$\bs{\hat{C}}_\tau$ and $\bs{N}_\tau$ are disjoint for all $\tau \in \bsq$. 
With regard to these sets we also state the following lemma, which we prove in the appendix:

\begin{lemma}\label{lemma:set sizes}
Let $\tau,\tau' \in \bsq$ and $\tau\neq\tau'$. The following equalities hold
\begin{align}\label{induc eq}
|\bs{N}_\tau \cap \bs{\hat{C}}_{\tau'}| = |\bs{\hat{C}}_\tau \cap \bs{\hat{C}}_{\tau'}|= |\bs{\hat{C}}_\tau \cap \bs{N}_{\tau'}| = |\bs{N}_\tau\cap \bs{N}_{\tau'}| = \frac{d^2}{4}.
\end{align}
Also for all $\tau \in \bsq $ we have
\begin{align}
|\bs{N}_{\sigma_0}\cap \bs{\hat{C}}_\tau| = |\bs{N}_{\sigma_0}\cap \bs{N}_\tau| = 0,\\
|\bs{\hat{C}}_{\sigma_0}\cap \bs{\hat{C}}_\tau| = |\bs{\hat{C}}_{\sigma_0}\cap \bs{N}_\tau| = \frac{d^2}{2}.
\end{align}
\end{lemma}

As mentioned above, $\bsqh$ forms an orthonormal basis for $\mc{M}_d$. We can define a representation $\varphi$ of the Clifford group by its action by conjugation on this basis, we have
\begin{equation}
\varphi: \C_q \to \mc{M}_d: C \mapsto \varphi(C)\sigma = C\sigma C\ct, \hspace{3mm} \sigma \in \bsqh,
\end{equation}
where $C$ is the standard representation of the Clifford group on $\md{C}^d$ discussed before. We call $\varphi$ the one-copy representation. Note that this representation is equivalent to the representation $C\otimes C^*$ where $*$ denotes the complex conjugate~\cite{Fulton2004}. It is a standard result~\cite{Wallman2014} that this representation decomposes into two irreducible subrepresentations carried by the spaces
\begin{equation}\label{eq:adjoint}
V_{\mathrm{id}} = \vsp\{\sigma_0\},\hspace{10mm} V_{\mathrm{adj}} = \vsp\{\;\sigma\;\;\|\;\; \sigma \in \bsq\}.
\end{equation}
The representation carried by the space $V_{\mathrm{adj}}$ is called the adjoint representation.\\

Note that we have for all $C\in \C_q$ and $\sigma \in \bsqh$ that $\varphi(C)\sigma  = \pm \tau$ for some $\tau\in \bsqh$. This means that in the basis $\bsqh$ the Clifford group is represented by signed permutation matrices. Note also that the action of the Clifford group through $\varphi$ is transitive on $\bsq$~\cite{Farinholt2014}. \\
Now we define analogously the two-copy representation of the $q$-qubit Clifford group 
$\C_q$ on the Hilbert space $\mc{M}_d\otimes \mc{M}_d = \mc{M}_d\tn{2}$ (hence the name two-copy representation). We define $\varphi\tn{2}$ with respect to its action on the basis
\begin{equation}\label{basis}
\bc{B} = \{\sigma_0\otimes \sigma_0, \sigma_0\otimes\sigma, 
\sigma\otimes\sigma_0, \sigma\otimes\tau\;\;\|\;\; \sigma,\tau\in\bsq\},
\end{equation}
of $\mc{M}_d\tn{2}$. We define the action of $\varphi\tn{2}$ on $\bc{B}$ as
\begin{equation}
\varphi\tn{2}(C)\sigma\otimes\tau = \big(C\sigma C^\dagger\big)\otimes\big(C\tau C^\dagger\big),\hspace{10mm}\sigma\otimes\tau \in \bc{B}.
\end{equation} 
Note that this representation is equivalent to the representation $C\otimes C^*\otimes C\otimes C^*$.
For brevity we will often forget about the tensor product symbol and write $\sigma\otimes \tau$ as $\sigma\tau$ when it is clear from the context.
Note that in the basis $\bc{B}$ the action of a Clifford element $C\in \C_q$ again takes the form of a signed permutation matrix. The rest of the paper will be concerned with identifying the irreducible subrepresentations of $\varphi\tn{2}$.\\ 

\section{The two-copy representation of the multi-qubit Clifford group}

The characterization 
for multiple qubits is more complicated than the single-qubit case considered 
previously~\cite{Wallman2014} because non-trivial elements of the multi-qubit 
Pauli group can commute, while others can 
anti-commute and these relations must be preserved 
under the action of the Clifford group~\cite{Farinholt2014}.
This section will be composed of several lemmas, ultimately culminating in \cref{theorem:two-copy representation}. In these lemmas we will introduce a variety of subspaces of $\mc{M}_d\tn{2}$ and prove that they all carry subrepresentations of $\varphi\tn{2}$. In \cref{theorem:two-copy representation} we will then exactly characterize which of the subspaces carry irreducible subrepresentations.
\renewcommand{\thefootnote}{\fnsymbol{footnote}}
We begin by calculating how many subrepresentations we require for each $q$.  The following lemma, proven in \cite{Zhu2015}, 
characterizes the inner product with itself of the 
character $\chi_{\varphi\tn{2}}$ of the two-copy representation of the Clifford 
group.\footnote[1]{Technically the character inner product of the representation $C\tn{4}$ rather than $C\!\otimes\! C^*\!\otimes \!C\!\otimes \!C^*$ is calculated in~\cite{Zhu2015}, but it can be easily seen that the character inner product is invariant under complex conjugation of some or all tensor factors of the representation.}

\begin{lemma}\label{lemma:upper bound}
Let $\C_q$ be the $q$-qubit Clifford group and $\varphi\tn{2}$ its two-copy 
representation with character $\chi_{\varphi\tn{2}}$. The character inner product of this representation with itself is 
\begin{align}
\inp{\chi_{\varphi\tn{2}}}{\chi_{\varphi\tn{2}}} = 
\begin{cases}{}
15 & \;\;q=1 \\
29 & \;\;q=2 \\
30 & \;\;q\geq 3.
\end{cases}
\end{align}
\end{lemma}
By \cref{Schur}, this number provides an upper limit to how many (in)equivalent 
irreducible subrepresentations the representation $\varphi\tn{2}$ can contain. We will now, over the course of several lemmas (\cref{lemma:adjoint representations,lemma:antisymmetric sector,lemma:symmetric commuting sector,lemma:symmetric anti-commuting sector,lemma:diagonal sector}), divide the space $\mc{M}_d\tn{2}$ into subspaces carrying subrepresentations of $\varphi\tn{2}$. This will eventually culminate in \cref{theorem:two-copy representation} where we prove that all the subrepresentations derived in  \cref{lemma:adjoint representations,lemma:antisymmetric sector,lemma:symmetric commuting sector,lemma:symmetric anti-commuting sector,lemma:diagonal sector} are in fact irreducible.\\

We continue by defining subspaces of the space $\mc{M}_d\tn{2}$ (spanned by $\bc{B}$) that carry subrepresentations of $\C_q$. Not all of these spaces will carry irreducible representations, these will be divided further in \cref{lemma:adjoint representations,lemma:antisymmetric sector,lemma:symmetric commuting sector,lemma:symmetric anti-commuting sector,lemma:diagonal sector}.
\begin{definition}\label{spaces}
Let $\bc{B}$ be the basis for $\mc{M}_d\tn{2}$ as in 
\cref{basis} and define the vectors
\begin{align}
	A_{\sigma,\tau} &:= \frac{1}{\sqrt{2}}(\sigma\tau - \tau\sigma), \\
	S_{\sigma,\tau} &:= \frac{1}{\sqrt{2}}(\sigma\tau + \tau\sigma)
\end{align}	
for $\sigma, \tau \in\bsq$ and $\sigma\neq \tau$.
We define the following subspaces of $\mc{M}_d\tn{2}$:
\begin{align}
V_{\mathrm{id}} &:= \vsp\{\sigma_0\sigma_0\}, \tag{trivial}\\
V_{\mathrm{r}} &:= \vsp\{\sigma_0\tau\;\;\|\;\;\tau\in\bsq\}, \tag{right 
adjoint}\\
V_{\mathrm{l}} &:= \vsp\{\tau\sigma_0\;\;\|\;\;\tau\in\bsq\}, \tag{left 
adjoint}\\
V_{\mathrm{d}} &:= \vsp\{\tau\tau\;\;\|\;\;\tau\in\bsq\}, 
\tag{diagonal sector}\\
\Vsc &:= 
\vsp\Bigl\{S_{\sigma,\tau}\;\; \|\;\; \sigma\in \bs{C}_\tau,\;\; \tau \in \bsq\Bigr\},
\tag{symmetric commuting sector}\\
\Vsa &:= 
\vsp\Bigl\{S_{\sigma,\tau}\;\; \|\;\; \sigma\in \bs{N}_\tau,\;\; \tau \in \bsq\Bigr\},
\tag{symmetric anti-commuting sector}\\
\Vasc &:= 
\vsp\Bigl\{A_{\sigma,\tau}\;\; \|\;\; \sigma\in \bs{C}_\tau,\;\; \tau \in \bsq\Bigr\},
\tag{antisymmetric commuting sector}\\
\Vasa &:= 
\vsp\Bigl\{A_{\sigma,\tau}\;\; \|\;\; \sigma\in \bs{N}_\tau,\;\; \tau \in \bsq\Bigr\}.
\tag{antisymmetric anti-commuting sector}
\end{align}
\end{definition}

These spaces do not all carry irreducible subrepresentations of $\varphi\tn{2}$ but they do all carry subrepresentations. This is proven in the following lemma:

\begin{lemma}\label{lemma:subrepresentations}
	All spaces $W$ defined in \cref{spaces} carry a subrepresentation of the representation $\varphi\tn{2}$ of the Clifford 
	group $\C_q$, that is
\begin{equation}
\varphi\tn{2}(C)W \subset W \hspace{5mm} \forall C\in \C_q.
\end{equation}
Note that $W$ may be empty for $q=1$, in which case the statement 
holds trivially.
\end{lemma}

\begin{proof}
First note that $C\sigma_0 C\ct = \sigma_0$ for all $C\in \C_q$ and that for any $C\in \C_q$ and $\sigma\in \bsq$ there exists a $\tau\in \bsq $ such that $ C\sigma C\ct = \pm \tau$. This means the spaces $V_{\mathrm{id}}, V_\mathrm{r},V_\mathrm{l}$ and $V_\mathrm{d}$ carry a subrepresentation of $\varphi\tn{2}$. Note also that we have 
\begin{align}
\varphi\tn{2}(C)S_{\sigma,\tau} = S_{C\sigma C\ct, C\tau C\ct}\hspace{5mm}\hspace{2mm} C\in \C_q,\\
\varphi\tn{2}(C)A_{\sigma,\tau} = A_{C\sigma C\ct, C\tau C\ct}\hspace{5mm}\hspace{2mm} C\in \C_q,
\end{align}
for all $\sigma, \tau \in\bsq$ and $\sigma\neq \tau$ and also
\begin{align}
\{C\sigma C\ct, C\tau C\ct\} =0    &\iff \{\sigma, \tau\} = 0 \hspace{5mm}\hspace{2mm} C\in \C_q,\\
[C\sigma C\ct, C\tau C\ct] =0    &\iff [\sigma, \tau] = 0 \hspace{5mm}\hspace{3.3mm} C\in \C_q
\end{align}
for all $\sigma, \tau \in\bsqh$. From these equations it is easy to see that $\Vsc, \Vsa, \Vasc $ and $\Vasa$ carry subrepresentations of $\varphi\tn{2}$ as well. 
\end{proof} 
Note that since $V_{\mathrm{id}}$ is a trivial representation it is automatically irreducible. Over the next few lemmas we will further characterize the other spaces defined in \cref{spaces}, beginning with the diagonal sector, i.e. the space $V_\mathrm{d}$.

\begin{lemma}[Diagonal sector]\label{lemma:diagonal sector}
Take the space $V_\mathrm{d}$ as defined in \cref{spaces} and define the following $3$ subspaces 
\begin{align}
V_0 &:= \vsp\left\{w\in V_\mathrm{d} \;\;\|\;\;w = \frac{1}{\sqrt{d^2-1}}\sum_{\sigma\in \bsq} \sigma\sigma\right\}\\
V_{1} &:=\vsp\left\{v \in V_\mathrm{d}\;\;\|\;\; v = \sum_{\sigma\in \bsq}\lambda_\sigma \sigma\sigma,\;\; \sum_{\sigma\in \bsq}\lambda_\sigma = 0,\;\; \sum_{\sigma\in \bs{N}_\tau}\lambda_\sigma = -\frac{d}{2}\lambda_\tau,\;\; \forall \tau\in \bsq \right\}\\
V_{2}&:= \vsp\left\{v \in V_\mathrm{d}\;\;\|\;\; v = \sum_{\sigma\in\bsq}\lambda_\sigma\sigma\sigma,\;\; \sum_{\sigma\in \bsq}\lambda_\sigma = 0,\;\; \sum_{\sigma\in \bs{N}_\tau}\!\!\lambda_\sigma  = \frac{d}{2}\lambda_\tau,\;\; \forall \tau\in \bsq\right\}
\end{align}
with $|V_1| = \frac{d(d+1)}{2} -1$ and $|V_2| = \frac{d(d-1)}{2} -1$. We have the following statements
\begin{itemize}
	\item For $q=1$ the spaces $V_0$ and $V_1$ carry irreducible subrepresentations of $\varphi\tn{2}$ and $V_2=0$.
	\item For $q \geq 2$ the spaces $V_0, V_1$ and $V_2$ carry irreducible subrepresentations of $\varphi\tn{2}$. 
\end{itemize}
\end{lemma}

\begin{proof}
The special case of $q=1$ was treated in \cite{Wallman2014}. We will treat the case $q\geq 2$.
We begin by establishing that the space $V_\mathrm{d} = \vsp\{\sigma\sigma\;\|\; \sigma\in \bsq\}$ has exactly three subspaces carrying
inequivalent subrepresentations of $\varphi\tn{2}$. One can see this by considering the character 
$\chi_d$ of $\varphi\tn{2}$ restricted to $V_\mathrm{d}$. It is easy to see by direct calculation that for all $C\in \C_q$ 
we have $\chi_\mathrm{d}(C) = F(C)$ where $F(C)$ is the number of non-identity Pauli 
matrices fixed under conjugation by $C$ up to a sign. This means the character inner product $\inp{\chi_\mathrm{d}}{\chi_\mathrm{d}}$ is given by
\begin{equation}
\inp{\chi_\mathrm{d}}{\chi_\mathrm{d}} = \frac{1}{|C|}\sum_{C\in \C_q} F(C)^2.
\end{equation}
By a generalized version of Burnside's Lemma (see~\cite{Zhu2015}) we can relate this to the number of orbits of the Clifford group (up to signs) on the set $\bsq \times \bsq$. These orbits were characterized in~\cite{Zhu2015} which yielded $\inp{\chi_\mathrm{d}}{\chi_\mathrm{d}} =3$ for $q\geq2$. This means, by \cref{Schur}, that 
$V_\mathrm{d}$ must contain exactly three inequivalent irreducible subrepresentations (all with multiplicity one). It is easy to see that $V_0$ carries a trivial subrepresentation by noting that $\varphi\tn{2}$ acts as a permutation on the basis of $V_\mathrm{d}$. Hence we can write $V_\mathrm{d} = V_0\oplus V_{\mathrm{orth}}$ where
\begin{equation}
V_{\mathrm{orth}} := \vsp\left\{v\in V_{\mathrm{d}}\;\;\|\;\; v = \sum_{\sigma\in \bsq}\lambda_\sigma \sigma\sigma,\;\;\sum_{\sigma\in \bsq}\lambda_{\sigma}=0\right\}.
\end{equation}
Because of the character argument given above we know this space must decompose into exactly two orthogonal subspaces $V_1,V_2$ which carry irreducible inequivalent subrepresentations of $\varphi\tn{2}$. We now characterize these subrepresentations.
We define the linear map $\bc{T}:V_\mathrm{d}\to V_\mathrm{d}$ by its action on the basis of $V_\mathrm{d}$. For all $\tau \in \bsq$ we have
\begin{align}
\bc{T}(\tau\tau) &:= \sum_{\sigma \in \bs{N}_\tau}\sigma\sigma.
\end{align}
It is easy to see that this map commutes with the action of $\varphi\tn{2}$ on $V_\mathrm{d}$. Hence, by the character argument above and Schur's lemma (\cref{eq:Schur corrolary}), it must be of the form
\begin{equation}
\bc{T} = a_0 P_0 + a_1 P_1 + a_2 P_2,
\end{equation}
where $P_0$ is the projector onto the space $V_0$ and $P_1,P_2$ are projectors onto the eigenspaces of $\bc{T}$ with eigenvalues $a_1,a_2$ respectively. We will label these eigenspaces $V_1$ and $V_2$. Note that $a_1,a_2\in \md{R}$ since $\bc{T}$ is symmetric. We will also assume that $a_1\leq a_2$. This can always be achieved by relabeling. By inspection we see that $a_0 = \frac{d^2}{2}$. We find can $a_1,a_2$ by considering the squared operator $\bc{T}^2$. We can compute its matrix elements in the given basis of $V_\mathrm{d}$ as
\begin{align}
\left[\bc{T}^2\right]_{\tau \tau'} = \inp{\tau\tau}{\mc{T}^2(\tau'\tau')}&= \sum_{\sigma\in \bs{N}_{\tau'}} \sum_{\hat{\sigma}\in \bs{N}_\sigma}\inp{\tau}{\hat{\sigma}}^2 \\
&=|\bs{N}_\tau\cap \bs{N}_{\tau'}|\\
&= \frac{d^2}{4} + \frac{d^2}{4}\delta_{\tau,\tau'},
\end{align}
where the last equality follows from \cref{lemma:set sizes} and $|\bs{N}_\tau| = \frac{d^2}{2}$ for all $\tau \in \bsq$.
From this characterization we can find the action of $\bc{T}^2$ on $v\in V_{\mathrm{orth}}$:
\begin{align}
\bc{T}^2 (v) &= \sum_{\sigma\in \bsq}\lambda_\sigma \bc{T}^2(\sigma\sigma)\\
&=\sum_{\sigma\in \bsq} \left(\lambda_\sigma \frac{d^2}{2} + \sum_{\hat{\sigma} \in \bsq\backslash\{\sigma\}} \lambda_{\hat{\sigma}}\frac{d^2}{4}\right)\sigma\sigma\\
& = \sum_{\sigma\in \bsq} \left(\lambda_\sigma \frac{d^2}{2} - \lambda_{\sigma}\frac{d^2}{4}\right)\sigma\sigma\\
& = \frac{d^2}{4} v,
\end{align}
where we used the definition of $v\in V_{\mathrm{orth}}$. This means that we must have $a_1^2 = a^2_2 = \frac{d^2}{4}$. There are hence two options: either $a_1 = a_2$ or $a_1 = -a_2$. We can exclude the first option by noting that the operator $\bc{T}$ is traceless. Hence we must have 
\begin{equation}
\tr(\bc{T}) = 0 = a_0 + a_1 |V_1| + a_2 |V_2| = \frac{d^2}{2}   + a_1 |V_1| + a_2 |V_2|,
\end{equation}
where $|V_i|$ is the dimension of the space $V_i$. By noting that $|V_1|+ |V_2| = d^2-2$ and that $V_1, V_2 \neq 0$ (this is a consequence of the character argument above) we find the only possible solution to be
\begin{align}
|V_1| &= \frac{d(d+1)}{2} -1, \hspace{10mm} a_1 = -\frac{d}{2},\\
|V_2| &= \frac{d(d-1)}{2} -1, \hspace{10mm} a_2 = \frac{d}{2}.
\end{align}
We can now diagonalize the operator $\bc{T}$ to find the description for the spaces $V_1,V_2$ given in the lemma statement.

\end{proof}
Next we establish an equivalence between the representations carried by $V_\mathrm{r}$ and $V_\mathrm{l}$ and two subspaces in the symmetric and antisymmetric sectors. All four of these spaces will be equivalent to the adjoint representation of the Clifford group, already mentioned in \cref{eq:adjoint}.
\begin{lemma}[Adjoint representations]\label{lemma:adjoint representations}
Take the vector spaces $V_\mathrm{r},V_\mathrm{l}$ as defined in \cref{spaces}. Also define the vector spaces 
	\begin{align}
	\Vadjc &:= 
\vsp\Bigl\{v^{[\mathrm{adj}]}_\tau\in \Vsc\;\;\|\;\; v^{[\mathrm{adj}]}_\tau = \frac{1}{\sqrt{2|\bs{C}_\tau|}} \sum_{\sigma \in \bs{C}_\tau}S_{\sigma,\sigma\cdot\tau},\;\;\tau\in\bsq\Bigr\} 
\tag{symmetric adjoint}\\
\Vadja &:= \vsp\Bigl\{v^{\{\mathrm{adj}\}}_\tau\in \Vasa \;\;\|\;\;v^{\{\mathrm{adj}\}}_\tau= \frac{1}{\sqrt{2|\bs{N}_\tau|}} \sum_{\sigma\in \bs{N}_\tau}A_{\sigma,i\sigma\cdot\tau},\;\;\tau\in\bsq\Bigr\} 
\tag{antisymmetric adjoint}
\end{align}
located in the symmetric commuting and antisymmetric anti-commuting sectors. The spaces $V_\mathrm{r},V_\mathrm{l}, \Vadja$ and $\Vadjc$ carry equivalent irreducible representations.

\end{lemma}
\begin{proof}
Note that the representations carried by the spaces $V_\mathrm{r},V_\mathrm{l}$ are trivially equivalent to the adjoint  
representation~(\cref{eq:adjoint}) of the Clifford group, which is  
irreducible~\cite{DiVincenzo2002}. This leaves us with the spaces $\Vadja$ and $\Vadjc$.
We begin by noting that the spaces $\Vadjc, \Vadja$ carry subrepresentations. This is easily seen by by taking $v^{[\mathrm{adj}]}_\tau\in \Vadjc$ as defined in the lemma statement and writing
\begin{align}
\varphi\tn{2}(C)v^{[\mathrm{adj}]}_\tau &= \frac{1}{\sqrt{2|\bs{C}_\tau|}} \sum_{\sigma\in \bs{C}_\tau}S_{C\sigma C\ct,(C\sigma C\ct) \cdot (C\tau C\ct)}\\
 &= \frac{1}{\sqrt{2|\bs{C}_\tau|}} \sum_{C\ct\sigma C\in \bs{C}_\tau} S_{\sigma, \sigma  \cdot C\tau C\ct}\\
  &= \frac{1}{\sqrt{2|\bs{C}_\tau|}} \sum_{\sigma\in \bs{C}_{C\tau C\ct}} S_{\sigma, \sigma  \cdot C\tau C\ct}\\
  &= v^{[\mathrm{adj}]}_{C\tau C\ct} \in \Vadjc,
\end{align}
where we used the fact that the action of the Clifford group preserves commutativity of elements of the Pauli group and acts transitively on $\bsq$. We have a similar argument for $\Vadja$.
Note also that the vectors spanning $\Vadjc$ as given in the lemma statement form an orthonormal basis for $\Vadjc$. For $\tau,\tau'\in\bsq$ we have
\begin{align}
\inp{v^{[\mathrm{adj}]}_\tau}{v^{[\mathrm{adj}]}_{\tau'}}&=\frac{1}{2|\bs{C}_\tau|} \sum_{\sigma\in \bs{C}_\tau}\sum_{\hat{\sigma} \in  \bs{C}_{\tau'}} \inp{S_{\sigma, \sigma\cdot\tau}}{S_{\hat{\sigma}, \tau'\cdot\hat{\sigma}}}\\
&= \frac{1}{2|\bs{C}_\tau|} \sum_{\sigma\in \bs{C}_\tau}\sum_{\hat{\sigma} \in  \bs{C}_{\tau'}}\inp{\sigma}{\hat{\sigma}}\inp{\sigma\cdot\tau}{\hat{\sigma}\tau'} \notag\\&\hspace{20mm}+ \frac{1}{2|\bs{C}_\tau|} \sum_{\sigma\in \bs{C}_\tau}\sum_{\hat{\sigma} \in  \bs{C}_{\tau'}}\inp{\sigma \cdot\tau}{\hat{\sigma}}\inp{\sigma}{\tau'\cdot\hat{\sigma}}\\
&= \frac{1}{2|\bs{C}_\tau|} \sum_{\sigma\in \bs{C}_\tau}\sum_{\hat{\sigma} \in  \bs{C}_{\tau'}} \delta_{\tau,\tau'}\delta_{\sigma,\hat{\sigma}} +\delta_{\tau,\tau'} \frac{1}{2|\bs{C}_\tau|} \sum_{\sigma\in \bs{C}_\tau}\sum_{\hat{\sigma} \in  \bs{C}_{\tau}\cap \bs{C}_\sigma} \inp{\sigma\cdot\hat{\sigma}}{\tau}^2  \\
&= \frac{1}{2}\delta_{\tau, \tau'} + \delta_{\tau,\tau'}\frac{1}{|\bs{C}_\tau|} \sum_{\sigma\in \bs{C}_\tau}\sum_{\hat{\sigma}\in \bs{C}_\tau}\delta_{\sigma,\hat{\sigma}}\\
& = \delta_{\tau,\tau'},
\end{align}
where we obtained the second to last equality by using \cref{lemma:set sizes} and noting that $\inp{\sigma\cdot\sigma_0}{\tau}=0$ if $\sigma\in C_\tau$. We can make a similar argument for the vectors spanning $\Vadja$. Now since $|V_{\mathrm{r}}| = |\Vadjc|$ we can construct the isomorphism 
\begin{equation}
\theta:V_\mathrm{r}\to \Vadjc:\sigma_0\tau \mapsto v_\tau^{[\mathrm{adj}]}.
\end{equation}
We can check that this isomorphism commutes with the action of $\varphi\tn{2}$. We have for all $\tau \in \bsq$
\begin{align}
\theta(\varphi\tn{2}(C)\sigma_0\tau) &= \frac{1}{\sqrt{2|\bs{C}_\tau|}}\sum_{\sigma\in \bs{C}_{C\tau C\ct}} S_{\sigma, \sigma \cdot C\tau C\ct}\\
&=\frac{1}{\sqrt{2|\bs{C}_\tau|}}\sum_{C\ct\sigma C\in \bs{C}_{\tau}} S_{\sigma, \sigma \cdot C\tau C\ct}\\
&=\frac{1}{\sqrt{2|\bs{C}_\tau|}}\sum_{\sigma\in \bs{C}_{\tau }} S_{C \sigma C\ct, C\sigma C\ct\cdot C\tau C\ct}\\
&= \varphi\tn{2}(C)\theta(\sigma_0\tau),
\end{align}
for all $C\in \C_q$. This means that the spaces $V_{\mathrm{r}}$ and $\Vadjc$ carry equivalent subrepresentations of $\varphi\tn{2}$. We can make the same argument for 
$\Vadja$ and hence $\Vadjc, \Vadja,V_\mathrm{r},V_\mathrm{l}$ carry equivalent 
irreducible representations.
\end{proof}
Now we turn our attention to the antisymmetric sector, i.e. the spaces $\Vasc, \Vasa$ as defined in \cref{spaces}, where we can formulate the following lemma.


\begin{lemma}[Antisymmetric sector]\label{lemma:antisymmetric sector}
Take the space $\Vasa$ as defined in \cref{spaces} and note that it contains the space $\Vadja$ (defined in \cref{lemma:adjoint representations}). Denote the orthogonal complement of $\Vadja$ in $\Vasa$ as $\Vadjap$. We have that the subrepresentations of $\varphi\tn{2}$ carried by $\Vasc$ and $\Vadjap$ are equivalent.
\end{lemma}
\begin{proof}
Note that $\Vadjap$ carries a subrepresentation of $\varphi\tn{2}$ by Maschke's lemma~\cite{Fulton2004} since $\Vadja$ and $\Vasa$ carry subrepresentations. We will prove that the representations carried by $\Vadjap$ and $\Vasc$ are equivalent by constructing an isomorphism between them that commutes with the action of $\varphi\tn{2}$. Note first that we can write down the following orthogonal basis for $\Vasc$ as
\begin{equation}
\Vasc = \vsp\{A_{\sigma,\sigma\cdot \tau} \;\;\|\;\; \sigma \in \bs{C}_\tau,\;\;\tau \in \bsq\}.
\end{equation}
Now consider the following linear map (defined as the linear extension of its action on the basis defined above) between $\Vasc$ and $\Vasa$.
\begin{equation}
\Theta: \Vasc \to \Vasa: A_{\sigma, \sigma \cdot \tau} \mapsto \sum_{\hat{\sigma}\in \bs{N}_\tau\cap \bs{C}_\sigma} A_{\hat{\sigma}, i\hat{\sigma}\cdot\tau} - \sum_{\hat{\sigma}\in \bs{N}_\tau\cap \bs{N}_\sigma} A_{\hat{\sigma}, i\hat{\sigma}\cdot\tau}
\end{equation}
for all $\sigma \in \bs{C}_\tau,\;\;\tau \in \bsq$. We now argue that the image of $\Theta$ is orthogonal to the space $\Vadja$. We do this by direct calculation. For all $\nu \in \bsq $ and all $\sigma \in \bs{C}_\tau,\;\;\tau \in \bsq$ we can calculate
\begin{align}
\sqrt{2|\bs{N}_\tau|}\inp{v^{\{\mathrm{adj}\}}_{\nu}}{\Theta\big(A_{\sigma,\sigma\cdot\tau}\big)}&= \sum_{\sigma'\in \bs{N}_\nu}\sum_{\hat{\sigma}\in \bs{N}_\tau\cap \bs{C}_\sigma} \inp{A_{\sigma',i\sigma'\cdot\nu}}{A_{\hat{\sigma},i\hat{\sigma}\cdot \tau}} - \sum_{\sigma'\in \bs{N}_\nu}\sum_{\hat{\sigma}\in \bs{N}_\tau\cap \bs{N}_\sigma} \inp{A_{\sigma',i\sigma'\cdot\nu}}{A_{\hat{\sigma},i\hat{\sigma}\cdot \tau}}\\
&=\sum_{\sigma'\in \bs{N}_\nu}\sum_{\hat{\sigma}\in \bs{N}_\tau\cap \bs{C}_\sigma}(\delta_{\sigma',\hat{\sigma}}+\delta_{\sigma',i\hat{\sigma}\cdot\tau}) \delta_{\nu,\tau} \notag\\
&\hspace{25mm}- \sum_{\sigma'\in \bs{N}_\nu}\sum_{\hat{\sigma}\in \bs{N}_\tau\cap \bs{N}_\sigma}(\delta_{\sigma',\hat{\sigma}}+\delta_{\sigma',i\hat{\sigma}\cdot\tau})\delta_{\nu,\tau}\\
&= 2\left(|\bs{N}_\nu \cap \bs{N}_{\tau}\cap \bs{C}_{\sigma}| - |\bs{N}_\nu\cap \bs{N}_{\tau}\cap \bs{N}_{\sigma}|\right)\delta_{\tau,\nu}\\
&= 2\left[|\bs{N}_{\tau}\cap \bs{C}_{\sigma}| - | \bs{N}_{\tau}\cap \bs{N}_{\sigma}|\right]\delta_{\tau,\nu}\\
&=0
\end{align}
where in the last line we used  \cref{lemma:set sizes} and $|\bs{N}_{\tau}\cap \bs{C}_{\sigma}| = |\bs{N}_{\tau}\cap \bs{\hat{C}}_{\sigma}|$ if $\sigma\in \bs{C}_\tau$. This means that $\mathrm{Im}(\Theta) \subset \Vadjap$. We now argue that $\mathrm{Im}(\Theta) = \Vadjap$. We first note that $|\Vadjap| = |\Vasc|$. Furthermore we can show that $\Theta$ preserves orthogonality under the Hilbert-Schmidt inner product and that $\mathrm{Ker}(\Theta)=0$. By direct calculation we have for all $\tau, \tau'\in \bsq$ and $\sigma\in \bs{C}_\tau, \sigma'\in \bs{C}_{\tau'}$
\begin{align}
\inp{\Theta\big(A_{\sigma,\sigma\cdot\tau}\big)}{\Theta\big(A_{\sigma',\sigma'\cdot\tau'}\big)}&= \sum_{\substack{\hat{\sigma}\in \bs{N}_\tau\cap \bs{C}_\sigma\\\hat{\sigma}'\in \bs{N}_{\tau'}\cap \bs{C}_{\sigma'}}} \inp{A_{\hat{\sigma},i\hat{\sigma}\cdot\tau}}{A_{\hat{\sigma}',i\hat{\sigma}'\cdot \tau'}} -\sum_{\substack{\hat{\sigma}\in \bs{N}_\tau\cap \bs{N}_\sigma\\\hat{\sigma}'\in \bs{N}_{\tau'}\cap \bs{C}_{\sigma'}}} \inp{A_{\hat{\sigma},i\hat{\sigma}\cdot\tau}}{A_{\hat{\sigma}',i\hat{\sigma}'\cdot \tau'}}\notag\\&\hspace{15mm} -\sum_{\substack{\hat{\sigma}\in \bs{N}_\tau\cap \bs{C}_\sigma\\\hat{\sigma}'\in \bs{N}_{\tau'}\cap \bs{N}_{\sigma'}}} \inp{A_{\hat{\sigma},i\hat{\sigma}\cdot\tau}}{A_{\hat{\sigma}',i\hat{\sigma}'\cdot \tau'}} + \sum_{\substack{\hat{\sigma}\in \bs{N}_\tau\cap \bs{N}_\sigma\\\hat{\sigma}'\in \bs{N}_{\tau'}\cap \bs{N}_{\sigma'}}} \inp{A_{\hat{\sigma},i\hat{\sigma}\cdot\tau}}{A_{\hat{\sigma}',i\hat{\sigma}'\cdot \tau'}}\\
&= \sum_{\substack{\hat{\sigma}\in \bs{N}_\tau\cap \bs{C}_\sigma\\\hat{\sigma}'\in \bs{N}_{\tau'}\cap \bs{C}_{\sigma'}}} (\delta_{\hat{\sigma}',\hat{\sigma}}+\delta_{\hat{\sigma}',i\hat{\sigma}\cdot\tau})\delta_{\tau,\tau'} -\sum_{\substack{\hat{\sigma}\in \bs{N}_\tau\cap \bs{N}_\sigma\\\hat{\sigma}'\in \bs{N}_{\tau'}\cap \bs{C}_{\sigma'}}}(\delta_{\hat{\sigma}',\hat{\sigma}}+\delta_{\hat{\sigma}',i\hat{\sigma}\cdot\tau})\delta_{\tau,\tau'} \notag\\&\hspace{10mm}-\sum_{\substack{\hat{\sigma}\in \bs{N}_\tau\cap \bs{C}_\sigma\\\hat{\sigma}'\in \bs{N}_{\tau'}\cap \bs{N}_{\sigma'}}} (\delta_{\hat{\sigma}',\hat{\sigma}}+\delta_{\hat{\sigma}',i\hat{\sigma}\cdot\tau})\delta_{\tau,\tau'} + \sum_{\substack{\hat{\sigma}\in \bs{N}_\tau\cap \bs{N}_\sigma\\\hat{\sigma}'\in \bs{N}_{\tau'}\cap \bs{N}_{\sigma'}}} (\delta_{\hat{\sigma}',\hat{\sigma}}+\delta_{\hat{\sigma}',i\hat{\sigma}\cdot\tau})\delta_{\tau,\tau'}\\
&= \bigg(|\bs{N}_\tau \cap \bs{C}_{\sigma}\cap \bs{C}_{\sigma'}| - |\bs{N}_\tau \cap \bs{C}_{\sigma}\cap \bs{N}_{\sigma'}|\notag\\&\hspace{30mm} -|\bs{N}_\tau \cap \bs{N}_{\sigma}\cap \bs{C}_{\sigma'}| +|\bs{N}_\tau \cap \bs{N}_{\sigma}\cap \bs{N}_{\sigma'}|\bigg)\delta_{\tau,\tau'}\label{eq:four_sets}
\end{align}
To further evaluate this expression we use the following fact. Let $\nu\in \bsqh$ such that $\sigma\cdot \sigma' \propto \nu$ (note that this implies that $\nu \in \bs{\hat{C}}_\tau$) . We then have
\begin{align}
\forall \mu \in \bsq:\;\;\;\; \mu \in \bs{C}_{\nu}\iff \mu \in (\bs{C}_\sigma\cap \bs{C}_{\sigma'}) \cup (\bs{N}_\sigma\cap \bs{N}_{\sigma'})\\
\forall \mu \in \bsq:\;\;\;\; \mu \in \bs{N}_{\nu}\iff \mu \in (\bs{C}_\sigma\cap \bs{N}_{\sigma'}) \cup (\bs{C}_\sigma\cap \bs{N}_{\sigma'}).
\end{align}
We use this together with the fact that $\bs{C}_\tau\cap \bs{N}_\tau =\emptyset$ for all $\tau \in \bsq$ to reduce \cref{eq:four_sets} to
\begin{align}
\left(|\bs{N}_\tau \cap \bs{C}_{\nu}| -|\bs{N}_\tau \cap \bs{N}_{\nu}| \right)\delta_{\tau,\tau'} = \left(\frac{d^2}{2}-1\right)\delta_{\tau,\tau'}(\delta_{\sigma,\sigma'}+ \delta_{\sigma,i\sigma'\cdot \tau})
\end{align}
where in the last equality we used \cref{lemma:set sizes} together with $\sigma\cdot \sigma' \propto \nu$ and that $\bs{N}_\tau\cap \bs{\hat{C}}_\nu = \bs{N}_\tau\cap \bs{C}_\nu$ if $\nu\in \bs{C}_\tau$ and that $\bs{C}_\nu = \bsq$ if $\nu=\sigma_0$ which occurs if and only if $\sigma = \sigma'$. Since $\inp{A_{\sigma,i\sigma\cdot\tau}}{A_{\sigma',\sigma'\cdot\tau'}} = \delta_{\tau,\tau'}(\delta_{\sigma,\sigma'}+ \delta_{\sigma,i\sigma'\cdot \tau}) $ this means that $\Theta$ preserves orthogonality and that $\mathrm{Ker}(\Theta)=0$. Together with the fact that $|\Vadjap| = |\Vasc|$ this implies that $\mathrm{Im}(\Theta)=\Vadjap$. This means we can restrict $\Theta$ to an isomorphism from $\Vasc$ to $\Vadjap$. We will abuse notation and refer to this isomorphism as $\Theta$ as well.\\

To prove that the representations carried by $\Vasc$ and $\Vadjap$ are equivalent we now still have to argue that $\Theta$ commutes with $\varphi\tn{2}$. We can do this by direct calculation. For all $\tau\in \bsq$ and $\sigma \in \bs{C}_\tau$ and $C\in \C_q$ we have
\begin{align}
\Theta(\varphi\tn{2}(C)(A_{\sigma,\sigma\cdot\tau})) &= \Theta(A_{C\sigma C \ct,C\sigma C\ct \cdot C \tau C\ct})\\
& = \sum_{\hat{\sigma}\in \bs{N}_{C\tau C\ct}\cap \bs{C}_{C\sigma C\ct}} A_{\hat{\sigma}, i \hat{\sigma}\cdot C \tau C\ct} - \sum_{\hat{\sigma}\in \bs{N}_{C\tau C\ct}\cap \bs{N}_{C\sigma C\ct}} A_{\hat{\sigma}, i \hat{\sigma}\cdot C \tau C\ct}\\
&=\sum_{C\ct \hat{\sigma}C\in \bs{N}_{\tau}\cap \bs{C}_{\sigma}} A_{\hat{\sigma}, i \hat{\sigma}\cdot C \tau C\ct} - \sum_{C\ct\hat{\sigma}C\in \bs{N}_{\tau}\cap \bs{N}_{\sigma }} A_{\hat{\sigma}, i \hat{\sigma}\cdot C \tau C\ct}\\
&=\sum_{\hat{\sigma}\in \bs{N}_{\tau}\cap \bs{C}_{\sigma}} A_{C\hat{\sigma}C\ct, i C\hat{\sigma} C\ct \cdot C \tau C\ct} - \sum_{\hat{\sigma}\in \bs{N}_{\tau}\cap \bs{N}_{\sigma }} A_{C\hat{\sigma}C\ct, i C\hat{\sigma}C\ct\cdot C \tau C\ct}\\
&=\varphi\tn{2}(C) \left(\sum_{\hat{\sigma}\in \bs{N}_{\tau}\cap \bs{C}_{\sigma}} A_{\hat{\sigma}, i \hat{\sigma}  \cdot  \tau } - \sum_{\hat{\sigma}\in \bs{N}_{\tau}\cap \bs{N}_{\sigma }} A_{\hat{\sigma}, i \hat{\sigma}\cdot  \tau }\right)\\
&=\varphi\tn{2}(C) \left(\Theta(A_{\sigma, \sigma\cdot \tau})\right).
\end{align}
This proves the equivalence of the subrepresentations carried by $\Vasc$ and $\Vadjap$.
 \end{proof}
Note that we have not proven that the subrepresentations carried by $\Vasc$~and $\Vadjap$ are irreducible. We will get this irreducibility for free when proving \cref{theorem:two-copy representation}.\\

Next up are the symmetric sectors. In order to facilitate the analysis of these spaces we begin by proving the following technical lemma. This technical lemma allows us to draw conclusions about the subrepresentations of $\varphi\tn{2}$ carried by subspaces of $\Vsc$ and $\Vsa$ by considering the action of a strict subgroup of the Clifford group $\C_q$ on particular subspaces of $\Vsc$ and $\Vsa$.

\begin{lemma}[space reduction]\label{lemma:space reduction}
For every $\tau\in \bsq$ define a subgroup $\C_q^\tau$ of $\C_q$ as
\begin{equation}
\C_q^\tau:= \{C\in \C_q\;\;\|\;\; C\tau C\ct = \pm \tau\}.
\end{equation}
Also define subspace V = $\vsp\{\;\sigma\hat{\sigma} \;\;\|\;\; \sigma, \hat{\sigma} \in \bsq,\;\; \sigma\neq\hat{\sigma}\}\subset \mc{M}_d\tn{2}$ and for all $\tau \in \bsq$ define the subspace 
\begin{equation}
V^\tau := \vsp\{\;\sigma_\tau \hat{\sigma}_\tau\;\;\|\;\; \sigma_\tau,\hat{\sigma}_\tau \in\bsq,\;\;\sigma_\tau\cdot \hat{\sigma}_\tau\propto \tau\}.
\end{equation}
The first claim of the lemma is:
\begin{itemize}
	\item The space $V$ decomposes with respect to $V^\tau$, that is
\begin{equation} 
V = \bigoplus_{\tau'\in \bsq}V^{\tau'}.
\end{equation}
\end{itemize}
Now assume that for some $\tau \in \bsq$ there exists a subspace $W^\tau$ of $V^\tau$ such that 
\begin{equation}
\varphi\tn{2}(\hat{C}) W^\tau \subset W^\tau, \;\;\;\;\; \forall \hat{C}\in \C_q^\tau.
\end{equation}
The second claim of the lemma is:
\begin{itemize}
\item For all $\tau'\in \bsq$ there exist $W^{\tau'} \subset V^{\tau'}$ such that $W^\tau$ and $W^{\tau'}$ are isomorphic and that
\begin{equation}
\varphi\tn{2}(C) W \subset W, \;\;\;\;\; \forall C\in \C_q,
\end{equation}
with 
\begin{equation}
W:= \bigoplus_{\tau'\in \bsq} W^{\tau'}.
\end{equation}
\end{itemize}
\end{lemma}
\begin{proof}
Note first that $\cup_{\tau' \in \bsq}V^{\tau'} = V$ and also for $\tau, \tau'\in \bs{\sigma}_q$ we have for all $\sigma_\tau \hat{\sigma}_\tau \in V^\tau$, $\sigma_{\tau'} \hat{\sigma}_{\tau'} \in V^{\tau'}$ that
\begin{equation}
\inp{\sigma_\tau \hat{\sigma}_\tau}{\sigma_{\tau'} \hat{\sigma}_{\tau'}} = \delta_{\sigma_\tau,\sigma_{\tau'}}  \delta_{\hat{\sigma}_\tau,\hat{\sigma}_{\tau'}}  = \delta_{\sigma_\tau,\sigma_{\tau'}}\delta_{\tau,\tau'},
\end{equation}
since if $\sigma_\tau = \sigma_{\tau'}$ we must have $(\hat{\sigma}_\tau = \hat{\sigma}_{\tau'} \iff \tau = \tau')$. This immediately implies
\begin{equation}
V = \bigoplus_{\tau'\in \bsq}V^{\tau'}.
\end{equation}
This proves the first claim of the lemma.\\

Now assume that there exists a $\tau \in \bsq$ such that there is a subspace $W^\tau\subset V^\tau$ such that for all $\hat{C}\in \C_q^\tau$ we have $\varphi\tn{2}(\hat{C})W^\tau\subset W^\tau$. For all $\tau'\in \bsq$ we can define the following subset $S_{\tau'}$ of $\C_q$:
\begin{equation}
S_{\tau'}:=\{C\in \C_q\;\;\|\;\; C\tau C\ct = \pm \tau'\}.
\end{equation}
Because the $\C_q$ acts transitively on $\bsq$ this set is never empty.
Now for every $C\in S_{\tau'}$ we can define the subspace $W^{C}$ as
\begin{equation}
W^{C,\tau'}:= \{\varphi\tn{2}(C)v\;\;\|\;\;v\in W^\tau\}.
\end{equation}
Note that for every $C\in S_{\tau'}$ we have $W^{C,\tau'}\subset V^{\tau'}$. We also have for $C_1,C_2 \in S_{\tau'}$ that 
\begin{align}
C_1\ct C_2 &\in \C_q^\tau,\\
C_2\ct C_1 &\in \C_q^\tau.
\end{align}
The first equation implies that 
\begin{equation}
\varphi\tn{2}(C_1\ct)\varphi\tn{2}(C_2)W^\tau \subset W^\tau \implies \varphi\tn{2}(C_1\ct)W^{C_2,\tau'} \subset W^\tau,
\end{equation}
which we can left-multiply by $\varphi\tn{2}(C_1)$ to get 
\begin{equation}
W^{C_2,\tau'}\subset W^{C_1,\tau'}.
\end{equation}
We can repeat this reasoning with $C_2,C_1$ interchanged to obtain $W^{C_2,\tau'}\subset W^{C_1,\tau'}$ and thus
\begin{equation}
W^{C_1,\tau'} = W^{C_2,\tau'},\;\;\;\; \forall C_1,C_2 \in S_{\tau'},
\end{equation}
for all $\tau'$. Let us label this single subspace by $W^{\tau'}$. Note that since $W^\tau\subset V^\tau$ for all $\tau \in \bsq$ the spaces $W^\tau, W^{\tau'}$ are orthogonal for $\tau \neq \tau'$. Hence we can consider the space
\begin{equation}
W = \bigoplus_{\tau\in \bsq} W^\tau.
\end{equation}
Now take $w \in W$. We can write
\begin{equation}
w = \sum_{\tau\in \bsq} v^\tau,\;\;\;\;\; v^\tau \in W^\tau.
\end{equation}
Now for all $C\in \bs{C}_q$ and $\tau \in \bsq$ there exist unique vectors $u^{\tau'}\in V^{\tau'}$ with $\tau'  = \pm C\tau C\ct$ such that
\begin{equation}
\varphi\tn{2}(C)w = \sum_{\tau\in \bsq} \varphi\tn{2}(C)v^\tau = \sum_{\tau'\in\bsq} u^{\tau'}\in W,
\end{equation}
which proves the lemma.

\end{proof}

Next we turn our attention to the symmetric commuting sector i.e., the space $\Vsc$.  We will decompose this space by using a curious connection between the representation $\varphi\tn{2}$ of $\C_q$ on $\Vsc$ and the representation $\varphi\tn{2}$ of $\C_{q-1}$ (the Clifford group on $q-1$ qubits) on its diagonal sector $V_\mathrm{d}^{q-1}$. We have the following lemma.
\begin{lemma}[Symmetric commuting sector]\label{lemma:symmetric commuting sector}
Take the space $\Vsc$ as defined in \cref{spaces}, the space $\Vadjc$ as defined in \cref{lemma:adjoint representations} and define the spaces 
\begin{align}
V_{[1]} := \bigoplus_{\tau\in \bsq} V_{[1]}^\tau, \;\;\;\;\;\;\;V_{[2]} := \bigoplus_{\tau\in \bsq} V_{[2]}^\tau,
\end{align}
where for all $\tau\in \bsq$
\begin{align}
V_{[1]}^\tau &:=\vsp\left\{v^\tau \in \Vsc\;\;\|\;\; v^\tau = \sum_{\sigma\in \bs{N}_\tau}\lambda_\sigma S_{\sigma,i\sigma\cdot \tau},\;\; \sum_{\sigma\in \bs{C}_\tau\cap \bs{N}_\nu}\!\!\!\!\lambda_\sigma = -\frac{d}{4}\lambda_\nu,\;\;\; \forall \nu\in \bs{C}_\tau \right\},\\
V_{[2]}^\tau&:= \vsp\left\{v^\tau \in \Vsc\;\;\|\;\; v^\tau = \sum_{\sigma\in \bs{N}_\tau}\lambda_\sigma S_{\sigma,i\sigma\cdot \tau},\;\; \sum_{\sigma\in \bs{C}_\tau\cap \bs{N}_\nu}\!\!\!\!\lambda_\sigma  = \frac{d}{4}\lambda_\nu,\;\;\; \forall \nu\in \bs{C}_\tau\right\}.
\end{align}
Note that $V_{[1]},V_{[2]}$ and $\Vadjc$ are subspaces of $\Vsc$. We have the following:
\begin{itemize}
	\item For $q=1$ we have $\Vsc = 0$ and hence $V_{[1]}=V_{[2]}=\Vadjc=0$.
	\item For $q=2$ we have $\Vsc =\Vadjc \oplus V_{[1]}$ and $V_{[2]} =0$. The spaces $\Vadjc $ and $ V_{[1]}$ carry irreducible subrepresentations of $\varphi\tn{2}$.
	\item For $q\geq3$ have $\Vsc =\Vadjc \oplus V_{[1]}\oplus V_{[2]}$. The spaces $\Vadjc,V_{[1]} $ and $ V_{[2]}$ carry irreducible subrepresentations of $\varphi\tn{2}$.
\end{itemize}
\end{lemma}
\begin{proof}
We begin the proof by noting that the statement is trivial for $q=1$. For $q\geq2$ we see that the space $\Vsc$ can be block decomposed in the following way
\begin{equation}
\Vsc = \bigoplus_{\tau\in \bsq} \Vsc^\tau,
\end{equation}
with 
\begin{equation}
\Vsc^{\tau} := \vsp\left\{ S_{\sigma, \sigma\cdot\tau}\;\;\|\;\;\sigma\in \bs{C}_\tau\right\},\;\;\;\;\;\;\;\;\tau\in \bsq.
\end{equation}
Using \cref{lemma:space reduction} we can, to find subspaces of $\Vsc$ carrying subrepresentations of $\varphi\tn{2}$, restrict ourselves to finding, for some $\tau \in \bsq$, subspaces of $\Vsc^\tau$ that are invariant under the representation $\varphi\tn{2}$ restricted the subgroup $\C_q^\tau\subset \C_q$ where $\C_q^\tau$ is defined as in \cref{lemma:space reduction}. For the purposes of this proof we choose $\tau =\sigma_Z\sigma_0$ where $\sigma_Z$ is the single qubit normalized Pauli element $Z/\sqrt{2}$. This means that we can write any $\hat{\sigma} \in \bs{C}_\tau$ as
\begin{equation}\label{eq:correspondence}
 \hat{\sigma} =\sigma_Z\sigma \hspace{2mm}\text{or}\hspace{2mm} \sigma_0\sigma,\hspace{3mm} \sigma \in \bs{\sigma}_{q-1},
\end{equation}
with $\bs{\sigma}_{q-1}$ the normalized, hermitian, non-identity Pauli elements on $q-1$ qubits. We also recall the definition of the diagonal sector on $q-1$ qubits:
\begin{equation}
V_\mathrm{d}^{q-1} := \vsp\{\;\sigma\sigma\;\;\|\;\;\;\sigma\in \bs{\sigma}_{q-1}\}.
\end{equation}
Since we have that
\begin{equation}\label{eq:symmetry}
S_{\sigma_0\sigma, \sigma_Z\sigma} = S_{\sigma_Z\sigma, \sigma_0\sigma}
\end{equation}
for all $\sigma\in \bs{\sigma}_{q-1}$ there is an isomorphism $\theta$ between the vector spaces $V_\mathrm{d}^{q-1}$ and $\Vsc^\tau$ of the form
\begin{equation}
\theta: V_\mathrm{d}^{q-1} \to \Vsc^\tau:\sigma\sigma \mapsto S_{\sigma_0\sigma, \sigma_Z\sigma}.
\end{equation}
Now consider the Clifford group on $q-1$ qubits, $\C_{q-1}$. It can be seen as a subgroup of the group $\C_q^\tau$ through the embedding
\begin{equation}
\hat{\theta}:\C_{q-1} \to \C_q^\tau: C \mapsto \id\otimes C.
\end{equation}
Now note that $\C_q^\tau$ preserves the commutation relations of the set $\bs{\sigma}_{q-1}$, that is, for all $\sigma, \hat{\sigma}\in \bs{\sigma}_{q-1}$ and $\sigma_1,\hat{\sigma}_1\in \{\sigma_0,\sigma_Z\}$ we have
\begin{equation}
\big[C(\sigma_1\sigma)C\ct,C(\hat{\sigma}_1\hat{\sigma})C\ct\big]=0\iff [\sigma_1\sigma,\hat{\sigma}_1\hat{\sigma}]=0 \iff [\sigma,\hat{\sigma}]=0
\end{equation}
for all $C\in \C_q^\tau$ with the same conclusion holding for the anti-commutator.
Now from this and \cref{eq:symmetry} one can see that for all $C\in \C_q^\tau$ there exists a $\hat{C} \in \C_{q-1}$ such that $\varphi\tn{2}(C)v =\varphi\tn{2}(\hat{\theta}(\hat{C}))v $ for all $v\in \Vsc^\tau$. This means that for any subspace $W$ of $\Vsc^\tau$ we have
\begin{equation}
\varphi\tn{2}(C)W\subset W,\;\; \forall C\in \C_q^\tau \iff \varphi\tn{2}(\hat{\theta}(\hat{C}))W\subset W,\;\; \forall\hat{C}\in \C_{q-1}
\end{equation}
Now let us consider the representation $\varphi\tn{2}$ of $\C_{q-1}$ on $q-1$ qubits. Let's label the restriction of this representation to $V_\mathrm{d}^{q-1}$ as $\varphi_\mathrm{d}$. From \cref{lemma:diagonal sector} we see that $V_\mathrm{d}^{q-1}$ for $q-1=1$ decomposes into two  spaces carrying irreducible subrepresentations of $\varphi_\mathrm{d}$ and for $q-1\geq 2$ decomposes into three spaces. We shall label these $V_0^{q-1}, V_1^{q-1}$ and $V_{2}^{q-1}$ with the tacit understanding that $V_{2}^{q-1} = 0$ for $q-1=1$. Now note that we have for all $\hat{C}\in\C_{q-1}$  and all $\sigma\in \bs{\sigma}_{q-1}$ that
\begin{equation}
\varphi\tn{2}(\hat{\theta}(\hat{C}))\theta(\sigma\sigma)  = S_{\sigma_0\hat{C}\sigma\hat{C}\ct, \sigma_Z\hat{C}\sigma\hat{C}\ct} = \theta(\varphi_\mathrm{d}(\hat{C})\sigma\sigma)
\end{equation}
which implies that the representations $\varphi_\mathrm{d}$ and the subrepresentation of $\varphi\tn{2}$ carried by $\Vsc^\tau$ restricted to the image of $\hat{\theta}$ are equivalent with the equivalence given by the map $\theta$. This means that the subspace $\Vsc^\tau$ (with $\tau=\sigma_Z\sigma_0$) decomposes into three subspaces carrying irreducible subrepresentations of $\varphi\tn{2}$ restricted to $\C^\tau_q$. We label these three spaces as
\begin{equation}
V_{[0]}^\tau:= \theta(V_0^{q-1}),\;\;\;\;\;V_{[1]}^\tau:= \theta(V_1^{q-1}),\;\;\;\;\;V_{[2]}^\tau:= \theta(V_2^{q-1}),
\end{equation}
with $\tau=\sigma_Z\sigma_0$. From \cref{lemma:space reduction} and identifying the spaces $\oplus_{\tau'\in \bsq}V_{[0]}^{\tau'}$ and $\Vadjc$ we now arrive at the lemma statement.
\end{proof}

Finally we analyze the symmetric anti-commuting sector, i.e the space $\Vsa$. This space carries an irreducible subrepresentation for $q=1$ and falls apart into two subspaces carrying irreducible subrepresentations for $q\geq2$. We have the following lemma.
\begin{lemma}[Symmetric anti-commuting sector]\label{lemma:symmetric anti-commuting sector}
Take the space $\Vsa$ as defined in \cref{spaces} and define the subspaces
\begin{align}
V_{\{1\}} := \bigoplus_{\tau\in \bsq} V_{\{1\}}^\tau,\;\; \;\;\;\;\;V_{\{2\}} := \bigoplus_{\tau\in \bsq} V_{\{2\}}^\tau,
\end{align}
where for all $\tau\in \bsq$:
\begin{align}
V_{\{1\}}^\tau &:=\vsp\left\{v^\tau \in \Vsa\;\;\|\;\; v^\tau = \sum_{\sigma\in \bs{N}_\tau}\lambda_\sigma S_{\sigma,i\sigma\cdot \tau},\;\; \sum_{\sigma\in \bs{N}_\tau\cap \bs{C}_\nu}\!\!\!\!\lambda_\sigma  - \sum_{\sigma\in \bs{N}_\tau\cap \bs{N}_\nu}\!\!\!\! \lambda_\sigma =\frac{d}{2}\lambda_\nu,\;\;\; \forall \nu\in \bs{N}_\tau \right\},\\
V_{\{2\}}^\tau&:= \vsp\left\{v^\tau \in \Vsa\;\;\|\;\; v^\tau = \sum_{\sigma\in \bs{N}_\tau}\lambda_\sigma S_{\sigma,i\sigma\cdot \tau},\;\; \sum_{\sigma\in \bs{N}_\tau\cap \bs{C}_\nu}\!\!\!\!\lambda_\sigma  - \sum_{\sigma\in \bs{N}_\tau\cap \bs{N}_\nu}\!\!\!\! \lambda_\sigma = -\frac{d}{2}2\lambda_\nu,\;\;\; \forall \nu\in \bs{N}_\tau\right\}.
\end{align}
We have the following statements:
\begin{itemize}
	\item For $q=1$ the space $\Vsa$ carries an irreducible subrepresentation of $\varphi\tn{2}$
	\item For $q\geq 2$ we have $\Vsa = V_{\{1\}}\oplus V_{\{2\}}$ and the spaces $V_{\{1\}}$ and $V_{\{2\}}$ carry subrepresentations of $\varphi\tn{2}$
\end{itemize}
\end{lemma}

\begin{proof}
The $q=1$ case was dealt with in~\cite{Wallman2014}, we will deal with the case of $q\geq 2$.
The argument goes by a combination of the arguments in \cref{lemma:diagonal sector} and \cref{lemma:symmetric commuting sector}. First note that we can write $\Vsa$ as 
\begin{equation}
\Vsa = \bigoplus_{\tau\in \bsq} \Vsa^\tau,\;\;\;\;\; \Vsa^\tau = \vsp\{S_{\sigma, i\sigma\cdot \tau}\;\;\|\;\;\sigma\in \bs{N}_\tau\}.
\end{equation}
We can again use \cref{lemma:space reduction} to look for subspaces of $\Vsa$ carrying subrepresentations of $\varphi\tn{2}$ by considering the action of the strict subgroup $\C_q^\tau$ of $\C_q$ on the space $\Vsa^\tau$ (where $\C_q^\tau$ is defined as in \cref{lemma:space reduction}). As in \cref{lemma:symmetric commuting sector} we choose $\tau = \sigma_Z \sigma_0$. The elements of $\bsq$ that anti-commute with $\tau$ can now be seen to be
\begin{equation}
\bs{N}_{\tau} = \{\;\sigma_X\sigma, \sigma_Y\sigma\;\;\|\;\;\sigma\in \bs{\hat{\sigma}}_{q-1}\},
\end{equation}
where $\sigma_X,\sigma_Y,\sigma_Z$ are again the normalized single qubit Pauli operators $X/\sqrt{2},Y/\sqrt{2},Z/\sqrt{2}$.
Note that the set $\bs{N}_{\tau}$ leads to an ambiguous definition of a basis for $\Vsa^\tau$ as we have that
\begin{equation}\label{eq:basis_sym_anti}
S_{\sigma_X\sigma,i (\sigma_X\sigma) \cdot (\sigma_Z\sigma_0)} = - S_{\sigma_Y\sigma,i (\sigma_Y\sigma) \cdot (\sigma_Z\sigma_0)}.
\end{equation}
for all $\sigma \in \hat{\bs{\sigma}}_{q-1}$ (recall that $\hat{\bs{\sigma}}_{q-1} = \bs{\sigma}_{q-1}\cup\{\sigma_0\}$).
We resolve this ambiguity by choosing the set $\{\sigma_X\sigma\;\;|\;\;\sigma\in \hat{\bs{\sigma}}_{q-1}\}$ to generate a basis of $\Vsa^\tau$. This makes that
\begin{equation}
\Vsa^\tau = \vsp \{S_{\sigma_X\sigma, \sigma_Y\sigma}\;\;\|\;\; \sigma\in \hat{\bs{\sigma}}_{q-1}\}.
\end{equation}
In the spirit of \cref{lemma:diagonal sector} we define the following linear map $\bc{A}$ as a linear extension of the action on the basis of $\Vsa^\tau$ as
\begin{equation}
\bc{A}\left(S_{\sigma_X\sigma, i(\sigma_X\sigma)\cdot \tau}\right) = \sum_{\hat{\sigma}\in \bs{C}_{\sigma_X\sigma}\cap \bs{N}_\tau} S_{\hat{\sigma},i\hat{\sigma}\cdot\tau} - \sum_{\hat{\sigma}\in \bs{N}_{\sigma_X\sigma} \cap \bs{N}_\tau} S_{\hat{\sigma},i\hat{\sigma}\cdot\tau}.
\end{equation}
We can argue that this map commutes with the action of $\varphi\tn{2}$ restricted to $\C_q^\tau$ (Where $\C_q^\tau$ is defined as in \cref{lemma:space reduction} with $\tau =\sigma_Z\sigma_0$) by direct calculation. We have for $C\in \C_q^\tau $ and $\sigma \in \bs{\hat{\sigma}}_{q-1}$:
\begin{align}
\bc{A}\left[\varphi\tn{2}(C)\left(S_{\sigma_X\sigma,i(\sigma_X\sigma)\cdot\tau}\right)\right] &= \bc{A}\left[S_{C(\sigma_X\sigma) C\ct,iC(\sigma_X\sigma)\cdot(\tau) C\ct}\right]\\
&=\sum_{\hat{\sigma}\in \bs{C}_{C(\sigma_X\sigma)C\ct}\cap \bs{N}_\tau} S_{\hat{\sigma},i\hat{\sigma}\cdot C\tau C\ct} - \sum_{\hat{\sigma}\in \bs{N}_{C(\sigma_X\sigma)C\ct} \cap \bs{N}_\tau} S_{\hat{\sigma},i\hat{\sigma}\cdot C\tau C\ct}\\
&=\sum_{C\ct\hat{\sigma}C\in \bs{C}_{\sigma_X\sigma}\cap \bs{N}_\tau} S_{\hat{\sigma},i\hat{\sigma}\cdot C\tau C\ct} - \sum_{C\ct\hat{\sigma}C\in \bs{N}_{\sigma_X\sigma} \cap \bs{N}_\tau} S_{\hat{\sigma},i\hat{\sigma}\cdot C\tau C\ct}\\
&=\sum_{\hat{\sigma}\in \bs{C}_{\sigma_X\sigma}\cap \bs{N}_\tau} S_{C\hat{\sigma}C\ct,iC\hat{\sigma}\cdot \tau C\ct} - \sum_{\hat{\sigma}\in \bs{N}_{\sigma_X\sigma} \cap \bs{N}_\tau} S_{C\hat{\sigma}C\ct,iC\hat{\sigma}\cdot \tau C\ct}\\
&= \varphi\tn{2}(C)\left(\bc{A}\left[S_{\sigma_X\sigma,i(\sigma_X\sigma)\cdot\tau}\right]\right).
\end{align}
This means that, through Schur's lemma the map $\bc{A}$ tells us something about the subrepresentations of $\varphi\tn{2}$ restricted to $\C_q^{\tau}$ carried by $\Vsa^\tau$. Because $\tau = \sigma_Z \sigma_0$ we can write $\bc{A}$ in a slightly better form by noting
\begin{align}
\bc{A}\left(S_{\sigma_X\sigma, i(\sigma_X\sigma)\cdot \tau}\right) &= \sum_{\hat{\sigma}\in \bs{C}_{\sigma_X\sigma}\cap \bs{N}_\tau} S_{\hat{\sigma},i\hat{\sigma}\cdot\tau} - \sum_{\hat{\sigma}\in \bs{N}_{\sigma_X\sigma} \cap \bs{N}_\tau} S_{\hat{\sigma},i\hat{\sigma}\cdot\tau}\\
&= \left[\sum_{\sigma'\in \bs{\hat{C}}_{\sigma}} S_{\sigma_X\sigma',\sigma_Y\sigma'} +  \sum_{\sigma'\in \bs{N}_{\sigma}} S_{\sigma_Y\sigma',-\sigma_X\sigma'}\right]  \notag\\&\hspace{30mm}- \left[\sum_{\sigma'\in \bs{N}_{\sigma}} S_{\sigma_X\sigma',\sigma_Y\sigma'} +  \sum_{\sigma'\in \bs{\hat{C}}_{\sigma}} S_{\sigma_Y\sigma',-\sigma_X\sigma'}\right]\\
&= 2 \left[\sum_{\sigma'\in \bs{\hat{C}}_{\sigma}} S_{\sigma_X\sigma',\sigma_Y\sigma'} -  \sum_{\sigma'\in \bs{N}_{\sigma}} S_{\sigma_X\sigma',\sigma_Y\sigma'}\right],
\end{align}
where we recall $\bs{\hat{C}}_\sigma$ to be $\bs{\hat{C}}_\sigma =\bs{C}_\sigma\cup \{\sigma_0,\sigma\}$.
We now analyze the properties of the map $\bc{A}$ by calculating $\tr(\bc{A})$ and $\bc{A}^2$. We have
\begin{align}
\frac{1}{2}\tr(\bc{A}) &=\frac{1}{2}\sum_{\sigma\in \bs{\hat{\sigma}}_{q-1}}\inp{S_{\sigma_X\sigma, \sigma_Y\sigma}}{\bc{A}\left(S_{\sigma_X\sigma, \sigma_Y\sigma}\right)} \\
&=\sum_{\sigma\in \bs{\hat{\sigma}}_{q-1}} \left[\sum_{\hat{\sigma}\in \bs{\hat{C}}_\sigma} \delta_{\sigma,\hat{\sigma}} - \sum_{\hat{\sigma}\in \bs{N}_\sigma} \delta_{\sigma,\hat{\sigma}}\right]\\
&=|\bs{\hat{\sigma}}_{q-1}| = \left(\frac{d}{2}\right)^2.
\end{align}
We can calculate $\bc{A}^2$ entry-wise. We abuse notation a little bit by denoting the entries of $\bc{A}^2$ as $[\bc{A}^2]_{\sigma,\hat{\sigma}}$ with $\sigma,\hat{\sigma} \in\bs{\hat{\sigma}}_{q-1}$ (this set has a one-to-one correspondence with the basis of $\Vsa$ in in \cref{eq:basis_sym_anti}). We calculate:
\begin{align}
\frac{1}{4}[\bc{A}^2]_{\sigma,\hat{\sigma}} &= \frac{1}{4}\inp{S_{\sigma_X \sigma,\sigma_Y \sigma}}{\bc{A}^2\left[S_{\sigma_X\hat{\sigma},\sigma_Y\hat{\sigma}}\right]} \\
&= \sum_{\substack{\sigma''\in \bs{\hat{C}}_{\sigma'}\\\sigma'\in \bs{\hat{C}}_\sigma}}\delta_{\sigma'',\hat{\sigma}} - \sum_{\substack{\sigma''\in \bs{N}_{\sigma'}\\\sigma'\in \bs{\hat{C}}_\sigma}}\delta_{\sigma'',\hat{\sigma}} - \sum_{\substack{\sigma''\in \bs{\hat{C}}_{\sigma'}\\\sigma'\in \bs{N}_\sigma}}\delta_{\sigma'',\hat{\sigma}} + \sum_{\substack{\sigma''\in \bs{N}_{\sigma'}\\\sigma'\in \bs{N}_\sigma}}\delta_{\sigma'',\hat{\sigma}}\\
&= |\bs{\hat{C}}_\sigma \cap \bs{\hat{C}}_{\hat{\sigma}}| - |\bs{\hat{C}}_\sigma \cap \bs{N}_{\hat{\sigma}}|- |\bs{N}_\sigma \cap \bs{\hat{C}}_{\hat{\sigma}}| + |\bs{N}_\sigma \cap \bs{N}_{\hat{\sigma}}|\\
&= \delta_{\sigma, \hat{\sigma}}|\bs{\hat{\sigma}}_{q-1}| = \delta_{\sigma, \hat{\sigma}}\left(\frac{d}{2}\right)^2,
\end{align}
where the last equality follows directly from \cref{lemma:set sizes}.
We see that $\bc{A}^2$ is proportional to the identity. This means that the eigenvalues of $\bc{A}$ must be $\pm d$. Since $\bc{A}$ is not proportional to the identity this means that both eigenvalues must be associated with non-trivial eigenspaces. Schur's lemma thus implies that $\Vsa^\tau$ carries a reducible subrepresentation of $\varphi\tn{2}$ restricted to $\C_q^\tau$ and moreover that the eigenspaces of $\bc{A}$ must be subrepresentations. We will call the spaces carrying these these subrepresentations $V^\tau_{\{1\}}$ and $V^\tau_{\{2\}}$ where we identify $V^\tau_{\{1\}}$ with the $d$ eigenvalue of $\bc{A}$ and $V^\tau_{\{2\}}$ with the $-d$ eigenvalue of $\bc{A}$. We can find out the dimensions of these spaces by noting that
\begin{align}
\tr(\bc{A})= d|V^\tau_{\{1\}}| - d|V^\tau_{\{2\}}| &= \frac{d^2}{2},\\
|V^\tau_{\{1\}}|+ |V^\tau_{\{2\}}| &=\left(\frac{d}{2}\right)^2.
\end{align}
Solving these equations yields
\begin{equation}
|V^\tau_{\{1\}}| = \frac{d}{4}\left(\frac{d}{2}+1\right),\;\;\;\;\;\;\; |V^\tau_{\{2\}}| = \frac{d}{4}\left(\frac{d}{2}-1\right).
\end{equation}
Diagonalizing $\bc{A}$ then yields the equations given in the lemma statement for $V^\tau_{\{1\}}$ and $V^\tau_{\{2\}}$ and by \cref{lemma:space reduction} we also get that $V_{\{1\}}$ and $V_{\{2\}}$ as defined in the lemma statement carry subrepresentations of the subrepresentation carried by $\Vsa$. 
\end{proof}
Note that we have not argued that the spaces $V_{\{1\}}, V_{\{2\}}$ carry irreducible subrepresentations. We will get the irreducibility for free in the full decomposition theorem, which we will deal with now. Using \cref{lemma:adjoint representations,lemma:antisymmetric sector,lemma:symmetric commuting sector,lemma:symmetric anti-commuting sector,lemma:diagonal sector} we can prove the main result of this paper: a decomposition of the two-copy representation $\varphi\tn{2}$ of the Clifford group $\C_q$ valid for any number of qubits $q$. We have:

\begin{theorem}[Decomposition of the two-copy representation]\label{theorem:two-copy representation}
The decomposition of the vector space $\mc{M}_d\tn{2} = \vsp\{\bc{B}\}$ into subspaces carrying 
irreducible subrepresentations of $\C_q$ in $\varphi\tn{2}$ for different 
values of $q$ is:
\begin{align*}
&\vspace{8mm}V_{\mathrm{id}}\oplus V_{\mathrm{r}}\oplus V_{\mathrm{l}} \oplus V_{0} \oplus V_{1} \oplus 
\Vsa\oplus \Vasa, \tag{$q=1$} \\
&\vspace{8mm}V_{\mathrm{id}}\oplus V_{{\mathrm{r}}}\oplus V_{\mathrm{l}}\oplus V_{0}\oplus V_{1} \oplus 
V_{2}\oplus \Vadjc\oplus V_{[1]} \oplus \Vadja\oplus V_{\{1\}}\oplus V_{\{2\}} \oplus 
\Vasc \oplus \Vadjap \tag{$q=2$}, \\
&\vspace{8mm}V_{\mathrm{id}}\oplus V_{\mathrm{r}}\oplus V_{\mathrm{l}}\oplus V_{0}\oplus V_{1} \oplus 
V_{2}\oplus \Vadjc\oplus V_{[1]}\oplus V_{[2]} \oplus \Vadja\oplus V_{\{1\}}\oplus V_{\{2\}} \oplus  \Vasc \oplus \Vadjap \tag{$q \geq 3$},
\end{align*}
where all spaces are as defined in \cref{spaces,lemma:adjoint representations,lemma:antisymmetric sector,lemma:symmetric commuting sector,lemma:symmetric anti-commuting sector,lemma:diagonal sector} and are gathered in \cref{table:irreducible subspaces} in the appendix. 
\end{theorem}

\begin{proof}
The $q=1$ case is dealt with in \cite{Wallman2014}. We will now deal with the cases $q=2$ and $q\geq 3$.
Beginning with $q\geq3$ note that we have already argued (in \cref{lemma:adjoint representations,lemma:antisymmetric sector,lemma:symmetric commuting sector,lemma:symmetric anti-commuting sector,lemma:diagonal sector}) that all spaces given in \cref{theorem:two-copy representation} are non-trivial and carry subrepresentations of $\varphi\tn{2}$. It remains to argue that these subrepresentations are all irreducible. We will do this using the Schur orthogonality relations (\cref{Schur}) and \cref{lemma:upper bound}. Begin by noting that the representations carried by the spaces $V_{\mathrm{r}},V_{\mathrm{l}},\Vadjc$ and $\Vadja$ are equivalent (\cref{lemma:adjoint representations}), the representations carried by the spaces $\Vasc$ and $\Vadjap$ are equivalent (\cref{lemma:antisymmetric sector}) and the representations carried by  $V_{\mathrm{id}}$ and $V_{0}$ are equivalent (Because they are both the trivial representation). Denote the character of the representations spanned by the direct sum of these representations by $\chi_{\mathrm{sum}}$. By the Schur orthogonality relations \cref{Schur} we have the following relation
\begin{equation}\label{eq:sum_chi}
\inp{\chi_{\mathrm{sum}}}{\chi_{\mathrm{sum}}} \geq 16 + 4 + 4 = 24,
\end{equation}
with equality if and only if all these spaces carry irreducible subrepresentations. Noting that we have yet to include the spaces $V_{1}, V_2, V_{\{1\}},V_{\{2\}},V_{[1]}$ and $V_{[2]}$ we can write the character of $\varphi\tn{2}$ as
\begin{equation}
\inp{\chi_{\varphi\tn{2}}}{\chi_{\varphi\tn{2}}} = \inp{\chi_{\mathrm{sum}} + \chi_{1}+\chi_{2}+ \chi_{\{1\}} +\chi_{\{2\}} + \chi_{[1]} + \chi_{[2]}}{\chi_{\mathrm{sum}} + \chi_{1}+\chi_{2}+ \chi_{\{1\}} +\chi_{\{2\}} + \chi_{[1]} + \chi_{[2]}},
\end{equation}
where $\chi_i$ is the character associated with the subrepresentation carried by the space $V_i$. Combining \cref{eq:sum_chi} and \cref{eq:Schur_ineq} we now conclude that
\begin{equation}
\inp{\chi_{\varphi\tn{2}}}{\chi_{\varphi\tn{2}}} \geq 30.
\end{equation}
From \cref{lemma:upper bound} we note that $\inp{\chi_{\varphi\tn{2}}}{\chi_{\varphi\tn{2}}} =30$ for $q\geq 3$. This means that all spaces mentioned must carry irreducible subrepresentations of $\varphi\tn{2}$ and that the spaces $V_{0}, V_\mathrm{r}, \Vasc,V_{1}, V_2, V_{\{1\}},V_{\{2\}},V_{[1]}$ and $V_{[2]}$ must carry mutually inequivalent irreducible representations. We can make the same argument for $q=2$ noting that the space $V_{[2]}=0$ (and hence does not contribute to the character inner product) and that for $q=2$ we have $\inp{\chi_{\varphi\tn{2}}}{\chi_{\varphi\tn{2}}} = 29$. This completes the classification of the 
irreducible representations of the two-copy representation $\varphi\tn{2}$ of the $q$-qubit Clifford group $\C_q$.
\end{proof}

\section{Conclusion}
We characterized the two-copy representation of the multi-qubit Clifford group 
and identified three distinct cases, namely, the single-qubit [analyzed 
in~\cite{Wallman2014}], two-qubit, and many-qubit cases, which contain $7$, $13$, 
and $14$ irreducible representations respectively. \\

As the Clifford group plays a central role in quantum information, we expect 
the present analysis to have many applications such as state \& channel tomography, analysis of fault-tolerance thresholds, large-deviation bounds \cite{Low2009} and state distinguishability (as analyzed in \cite{Gross2016,Gross2016b,Gross2016c}).
As a concrete example, we have used results from the present paper in a 
companion paper~\cite{Helsen2016} to provide a much sharper analysis of the 
statistical performance of randomized benchmarking~\cite{Knill2008,Magesan2011}. While this result advances understanding of the representation theory of the 
Clifford groups, there remain several open questions about general 
representation theory of multi-qubit Clifford groups. First and foremost, the 
character table of the Clifford group is unknown. Working out this table would greatly assist future studies. This 
paper has identified several distinct irreducible representations, which should 
assist in the construction of the character table. Finally, these results hold for qubits and 
generalizing them to higher-dimensional systems remains an open problem. \\

While writing the current results the authors became aware of an equivalent result due to Zhu, Kueng, Grassl and Gross,~\cite{Gross2016,Gross2016b,Gross2016c} where the fourth tensor power representation of the Clifford group (which is closely related to the two-copy representation) is analyzed using techniques from stabilizer codes and used to construct projective $4$-designs out of the orbits of the Clifford group, analyze POVM norm constants and applied to the problem of phase retrieval.
\subsection*{Acknowledgments}
We would like to thank Le Phuc Thinh, J\'er\'emy Ribeiro, Bas Dirkse and Axel Dahlberg for enlightening discussions and helpful comments. JH and SW 
are funded by STW Netherlands, NWO VIDI and an ERC Starting Grant. This 
research was supported by the U.S. Army Research Office through grant 
W911NF-14-1-0103.
\newpage
\bibliography{CliffordPaper-Final}


\appendix
\setcounter{lemma}{0}
\section{Proof of Lemma 1}
\begin{lemma}
Let $\tau,\tau' \in \bsq$ and $\tau\neq\tau'$. The following equalities hold
\begin{align}\label{induc eq app}
|\bs{N}_\tau \cap \bs{\hat{C}}_{\tau'}| = |\bs{\hat{C}}_\tau \cap \bs{\hat{C}}_{\tau'}|= |\bs{\hat{C}}_\tau \cap \bs{N}_{\tau'}| = |\bs{N}_\tau\cap \bs{N}_{\tau'}| = \frac{d^2}{4}.
\end{align}
Also for all $\tau \in \bsq $ we have
\begin{align}
|\bs{N}_{\sigma_0}\cap \bs{\hat{C}}_\tau| = |\bs{N}_{\sigma_0}\cap \bs{\hat{C}}_\tau| = 0,\\
|\bs{\hat{C}}_{\sigma_0}\cap \bs{\hat{C}}_\tau| = |\bs{\hat{C}}_{\sigma_0}\cap \bs{N}_\tau| = \frac{d^2}{2}.
\end{align}
\end{lemma}
\begin{proof}
Let $\tau, \tau'\in \bsq$ and $\tau \neq \tau'$. We begin by noting that $\bs{N}_\tau$ is the complement of $\bs{\hat{C}}_\tau$ in $\bsqh$ and that $|\bs{\hat{C}}_{\tau}| = |\bs{N}_\tau|= \frac{d^2}{2}$ for all $\tau \in \bsq$. This allows us to make the following statements
\begin{align}
|\bs{\hat{C}}_\tau\cap \bs{\hat{C}}_{\tau'}| + |\bs{N}_\tau\cap \bs{\hat{C}}_{\tau'}| = \frac{d^2}{2}, \hspace{10mm}|\bs{\hat{C}}_\tau\cap \bs{\hat{C}}_{\tau'}| + |\bs{\hat{C}}_\tau\cap \bs{N}_{\tau'}| = \frac{d^2}{2},\\
|\bs{N}_\tau\cap \bs{\hat{C}}_{\tau'}| + |\bs{N}_\tau\cap \bs{N}_{\tau'}| = \frac{d^2}{2},\hspace{10mm}|\bs{\hat{C}}_\tau\cap \bs{N}_{\tau'}| + |\bs{N}_\tau\cap \bs{N}_{\tau'}| = \frac{d^2}{2}.
\end{align}
We can solve this system of equations to obtain
\begin{align}
|\bs{\hat{C}}_\tau\cap \bs{\hat{C}}_{\tau'}| &=|\bs{N}_\tau\cap \bs{N}_{\tau'}|,\label{eq:system1}\\
|\bs{N}_\tau\cap \bs{\hat{C}}_{\tau'}|&=\frac{d^2}{2} - |\bs{N}_\tau\cap \bs{N}_{\tau'}|,\label{eq:system2}\\
|\bs{\hat{C}}_\tau\cap \bs{N}_{\tau'}|&=\frac{d^2}{2} - |\bs{N}_\tau\cap \bs{N}_{\tau'}|.\label{eq:system3}
\end{align}
The rest of the argument will proceed by induction on the number of qubits $q$ (recall that $d = 2^q$).
For $q=1$ we have that
\begin{equation}
|\bs{N}_\tau\cap \bs{N}_{\tau'}| = |\{\tau',i\tau\cdot\tau'\}\cap\{\tau,i\tau\cdot\tau'\}| = |\{i\tau\cdot \tau'\}| = 1 = \frac{2^2}{4}.
\end{equation}
From \cref{eq:system1,eq:system2,eq:system3} we then have that
\begin{equation}
|\bs{N}_\tau \cap \bs{\hat{C}}_{\tau'}| = |\bs{\hat{C}}_\tau \cap \bs{\hat{C}}_{\tau'}|= |\bs{N}_\tau \cap \bs{\hat{C}}_{\tau'}| = |\bs{N}_\tau\cap \bs{N}_{\tau'}| =1.
\end{equation}
Now assume \cref{induc eq app} to hold up to $q-1$. For $\tau\in \bsq$ we can write
\begin{equation}
\bs{N}_\tau = \left(\bs{N}_{\tau_1}\otimes \bs{\hat{C}}_{\tau_{q-1}} \right)\cup \left(\bs{\hat{C}}_{\tau_1}\otimes \bs{N}_{\tau_{q-1}} \right),\;\;\;\;\; \tau_1\in \hat{\bs{\sigma}}_1,\;\; \tau_{q-1} \in \bs{\sigma}_{q-1}, \;\;\; \text{s.t. }\tau_1\otimes \tau_{q-1} = \tau,
\end{equation}
where by $\bs{A}\otimes \bs{B}$ we mean $\bs{A}\otimes \bs{B} := \{a\otimes b \;\|\;a\in \bs{A},\;b\in\bs{B}\}$.
Now we can write
\begin{align}
|\bs{N}_\tau\cap \bs{N}_{\tau'}| &= \bigg|\left[\left(\bs{N}_{\tau_1}\otimes \bs{\hat{C}}_{\tau_{q-1}} \right)\cup \left(\bs{\hat{C}}_{\tau_1}\otimes \bs{N}_{\tau_{q-1}} \right)\right]\cap\left[\left(\bs{N}_{\tau'_1}\otimes \bs{\hat{C}}_{\tau'_{q-1}} \right)\cup \left(\bs{\hat{C}}_{\tau'_1}\otimes \bs{N}_{\tau'_{q-1}} \right)\right]\bigg|\\
&=\bigg|\left(\{\sigma_0\}\otimes(\bs{N}_{\tau_{q-1}}\cap\bs{N}_{\tau'_{q-1}})\right) \cup  \left(\{i\tau_1\cdot\tau'_1\}\otimes(\bs{\hat{C}}_{\tau_{q-1}}\cap\bs{\hat{C}}_{\tau'_{q-1}})\right)\notag\\
&\hspace{25mm}\cup \left(\{\tau_1\}\otimes(\bs{N}_{\tau_{q-1}}\cap \bs{\hat{C}}_{\tau'_{q-1}})\right)\cup \left(\{\tau'_1\}\otimes(\bs{\hat{C}}_{\tau_{q-1}}\cap \bs{N}_{\tau'_{q-1}})\right)\bigg|\\
&= \frac{d^2}{4},
\end{align}
where the last line holds by the induction hypothesis and the fact that all sets in the equation are disjoint. This proves the first half of the lemma. Now take $\tau\in \bsq$ and consider the sets $\bs{N}_{\sigma_0}, \bs{\hat{C}}_{\sigma_0}$. It is trivial to see that $\bs{N}_{\sigma_0} = \emptyset$ and $\bs{\hat{C}}_{\sigma_0} = \bsqh$. Since $|\bs{N}_\tau| = |\bs{\hat{C}}_\tau| = \frac{d^2}{2}$ the second half of the lemma also follows.
\end{proof}

\section{Table of all relevant vector spaces}

\begin{table}[h]
\centering
\ssmall
\def\arraystretch{2.5}
\makebox[\textwidth][c]{
\begin{tabular}{ C   C C  C }
\textbf{space}&		\textbf{definition}&					\textbf{irreducible}&			\textbf{dimension}\\\hline
V_{\mathrm{id}}&					\vsp\{\sigma_0\sigma_0\}&					q\geq1&							1\\
V_{\mathrm{r}}&					\vsp\{\sigma_0\tau\;\;\|\;\;\tau\in\bsq\}&					q\geq1&							d^2-1\\
V_{\mathrm{l}}&					\vsp\{\tau\sigma_0\;\;\|\;\;\tau\in\bsq\}&					q\geq1&							d^2-1\\
V_{\mathrm{d}}&					\vsp\{\tau\tau\;\;\|\;\;\tau\in\bsq\}&					\text{no}&						d^2-1\\
\Vsc&			\vsp\Bigl\{S_{\sigma,\tau}\;\; \|\;\; \sigma\in \bs{C}_\tau,\;\; \tau \in \bsq\Bigr\}&					\text{no}&						\frac{d^2-1}{2}\left(\frac{d^2}{2}-2 \right)\\
\Vsa&			\vsp\Bigl\{S_{\sigma,\tau}\;\; \|\;\; \sigma\in \bs{N}_\tau,\;\; \tau \in \bsq\Bigr\}&					q=1&							\frac{d^2-1}{2}\left(\frac{d^2}{2}\right)\\
\Vasc&			\vsp\Bigl\{A_{\sigma,\tau}\;\; \|\;\; \sigma\in \bs{C}_\tau,\;\; \tau \in \bsq\Bigr\}&					q\geq2&							\frac{d^2-1}{2}\left(\frac{d^2}{2}-2 \right)\\
\Vasa&			\vsp\Bigl\{A_{\sigma,\tau}\;\; \|\;\; \sigma\in \bs{N}_\tau,\;\; \tau \in \bsq\Bigr\}&					q=1&							\frac{d^2-1}{2}\left(\frac{d^2}{2} \right)\\
V_{0}&			\vsp\left\{w \in V_\mathrm{d}\;\;\|\;\; w = \frac{1}{\sqrt{d^2-1}}\sum_{\sigma\in \bsq} \sigma\sigma\right\}&					q\geq1&							1\\
V_{1}&			\vsp\left\{v \in V_d\;\;\|\;\; v = \sum_{\sigma\in \bsq}\lambda_\sigma \sigma\sigma,\;\; \sum_{\sigma\in \bsq}\lambda_\sigma = 0,\;\; \sum_{\sigma\in \bs{N}_\nu}\lambda_\sigma = -\frac{d}{2}\lambda_\tau,\;\; \tau\in \bsq \right\}&					q\geq1&				\frac{d(d+1)}{2}-1\\
V_{2}&			\vsp\left\{v \in V_d\;\;\|\;\; v = \sum_{\sigma\in\bsq}\lambda_\sigma\sigma\sigma,\;\; \sum_{\sigma\in \bsq}\lambda_\sigma = 0,\;\; \sum_{\sigma\in \bs{N}_\nu}\!\!\lambda_\sigma  = \frac{d}{2}\lambda_\tau,\;\; \tau\in \bsq\right\}&					q\geq2&				\frac{d(d-1)}{2}-1\\
\Vadjc&			\vsp\Bigl\{v^{[\mathrm{adj}]}_\tau\in \Vsc\;\;\|\;\;v^{[\mathrm{adj}]}_\tau= \frac{1}{\sqrt{2|\bs{C}_\tau|} }\sum_{\sigma\in \bs{C}_\tau}S_{\sigma,\sigma\cdot\tau},\;\;\tau\in\bsq\Bigr\}& 					q\geq2&							(d^2\!-\!1)\\
\Vadja&			\vsp\Bigl\{v^{\{\mathrm{adj}\}}_\tau \in \Vasa \;\;\|\;\;v^{\{\mathrm{adj}\}}_\tau= \frac{1}{\sqrt{2|\bs{N}_\tau|}} \sum_{\sigma\in \bs{N}_\tau}A_{\sigma,i\sigma\cdot\tau},\;\;\tau\in\bsq\Bigr\}&					q\geq2&							(d^2\!-\!1)\\
\Vadjap&		\vsp\Bigl\{v^{\{A\}} \in \Vasa\;\;\|\;\; \inp{v^{\{A\}}}{v^{\{\mathrm{adj}\}}} = 0,\;\; \forall \;v^{\{\mathrm{adj}\}} \in \Vadja\Bigr\}&					q\geq2&							(d^2\!-\!1)\left(\frac{d^2}{2}-2 \right)\\
V_{[1]}&		\vsp\left\{v^\tau \in \Vsc\;\;\|\;\; v^\tau = \sum_{\sigma\in  \bs{N}_\tau}\lambda_\sigma \sigma\sigma,\;\; \sum_{\sigma\in \bs{C}_\tau \cap \bs{N}_\nu}\lambda_\sigma = -d\lambda_\nu,\;\;\nu\in \bs{C}_\tau,\;\; \tau\in \bsq \right\}&					q\geq2&							(d^2\!-\!1)\left[\frac{\frac{d}{2}(\frac{d}{2}\!+\!1)}{2}\!-\!1\right]\\
V_{[2]}&		\vsp\left\{v^\tau \in \Vsc\;\;\|\;\; v^\tau = \sum_{\sigma\in \bs{N}_\tau}\lambda_\sigma\sigma\sigma,\;\; \sum_{\sigma\in\bs{C}_\tau\cap  \bs{N}_\nu}\!\!\lambda_\sigma  = d\lambda_\nu,\;\;\nu\in \bs{C}_\tau,\;\; \tau\in \bsq\right\}&					q\geq3&							(d^2\!-\!1)\left[\frac{\frac{d}{2}(\frac{d}{2}\!-\!1)}{2}\!-\!1\right]\\
V_{\{1\}}&		\vsp\left\{v^\tau \in \Vsa\;\;\|\;\; v^\tau = \sum_{\sigma\in \bs{N}_\tau}\lambda_\sigma S_{\sigma,i\sigma\cdot \tau},\;\; \sum_{\sigma\in \bs{N}_\tau\cap \bs{C}_\nu}\!\!\!\!\lambda_\sigma  -\!\!\!\!\! \sum_{\sigma\in \bs{N}_\tau\cap \bs{N}_\nu}\!\!\!\! \lambda_\sigma = \frac{d}{2}\lambda_\nu,\;\;  \nu\in \bs{N}_\tau,\;\; \tau \in \bsq \right\}&	q\geq1&							(d^2\!-\!1)\frac{\frac{d}{2}(\frac{d}{2}\!+\!1)}{2}\\
V_{\{2\}}&		\vsp\left\{v^\tau \in \Vsa\;\;\|\;\; v^\tau = \sum_{\sigma\in \bs{N}_\tau}\lambda_\sigma S_{\sigma,i\sigma\cdot \tau},\;\; \sum_{\sigma\in \bs{N}_\tau\cap \bs{C}_\nu}\!\!\!\!\lambda_\sigma  - \!\!\!\!\!\sum_{\sigma\in \bs{N}_\tau\cap \bs{N}_\nu}\!\!\!\! \lambda_\sigma = -\frac{d}{2}\lambda_\nu,\;\; \nu\in \bs{N}_\tau,\;\;\tau \in \bsq\right\}&					q\geq2&							(d^2\!-\!1)\frac{\frac{d}{2}(\frac{d}{2}\!-\!1)}{2}
\end{tabular}}
$\vspace{1em}$\\
\normalsize
\begin{tikzpicture}[sloped,level distance=1cm]
\node  {$\mc{M}_d\tn{2}$}
	child { 
			node {$V_{\mathrm{d}}$}
				child{
					node {$V_0$}
				}
				child{
					node {$V_1$}
				}
				child{
					node {$V_2$}
				}
		}
	child { 
		node {$V_{\mathrm{id}}$}
	}
	child { 
		node {$V_{\mathrm{r}}$}
	}
	child { 
		node {$\Vsc$}
			child {
					node {$\Vadjc$}
				}
			child {
					node {$V_{[1]}$}
				}
			child {
					node {$V_{[2]}$}
				}
	}
	child { 
		node {$V_{\mathrm{l}}$}
	}
	child { 
		node {$\Vsa$}
			child {
					node {$V_{\{1\}}$}
				}
			child {
					node {$V_{\{2\}}$}
				}
	}
	child { 
		node {$\Vasc$}
	}
	child { 
		node {$\Vasa$}
			child {
					node {$\Vadja$}
				}
			child {
					node {$\Vadjap$}
				}
	};
\end{tikzpicture}
\caption{Table with all subspaces of $\mc{M}_d$ carrying subrepresentations of $\varphi\tn{2}$. Given are the name in the text, the definition, for which values (if any) of $q\in \md{N}$ they carry irreducible subrepresentations of $\varphi\tn{2}$ and their dimension as a function of $d=2^q$. Also given is a tree diagram showing subspace inclusions where every child node is a subspace of its parent nodes. }\label{table:irreducible subspaces}
\end{table}
\end{document}